\documentclass[11pt]{article}

\usepackage{style}
\usepackage{hyperref}
\usepackage{subcaption}

\newcommand{\vecp}{\vec{p}}
\newcommand{\cE}{\mathcal{E}}

\DeclareMathOperator*{\amax}{argmax}

\newcommand{\ind}[1]{\mathbbm{1}\left\{ #1\right\}}
\newcommand{\cD}{\mathcal{D}}

\newcommand{\rew}{\text{reward}}
\DeclareMathAlphabet{\mathcal}{OMS}{zplm}{m}{n}

\newcommand{\cK}{\mathcal{K}}
\renewcommand{\vec}[1]{\bm{#1}}
\renewcommand{\P}{\mathbb{P}}

\title{Disincentivizing Polarization in Social Networks}

\author{
	Christian Borgs \\ UC Berkeley \\ \texttt{borgs@berkeley.edu}
 \and 
	Jennifer Chayes \\ UC Berkeley \\ \texttt{jchayes@berkeley.edu}
    \and
	Christian Ikeokwu \\ UC Berkeley \\ \texttt{christian\_ikeokwu@berkeley.edu}
  \and
	Ellen Vitercik \\ Stanford University \\ \texttt{vitercik@stanford.edu}}

\begin{document}

\maketitle

\begin{abstract}
On social networks, algorithmic personalization drives users into filter bubbles where they rarely see content that deviates from their interests. We present a model for content curation and personalization that avoids filter bubbles, along with algorithmic guarantees and nearly matching lower bounds. In our model, the platform interacts with $n$ users over $T$ timesteps, choosing content for each user from $k$ categories. The platform receives stochastic rewards as in a multi-arm bandit. To avoid filter bubbles, we draw on the intuition that if some users are shown some category of content, then all users should see at least a small amount of that content. We first analyze a naive formalization of this intuition and show it has unintended consequences: it leads to ``tyranny of the majority'' with the burden of diversification borne disproportionately by those with minority interests. This leads us to our model which distributes this burden more equitably. We require that the probability any user is shown a particular type of content is at least $\gamma$ times the average probability all users are shown that type of content. Full personalization corresponds to $\gamma = 0$ and complete homogenization corresponds to $\gamma = 1$; hence, $\gamma$ encodes a hard cap on the level of personalization. We also analyze additional formulations where the platform can exceed its cap but pays a penalty proportional to its constraint violation. We provide algorithmic guarantees for optimizing recommendations subject to these constraints. These include nearly matching upper and lower bounds for the entire range of $\gamma \in [0,1]$ showing that the cumulative reward of a multi-agent variant of the Upper-Confidence-Bound algorithm is nearly optimal. Using real-world preference data, we empirically verify that under our model, users share the burden of diversification and experience only minor utility loss when recommended more diversified content.
\end{abstract}

\section{Introduction}
Over the past decade, large internet platforms have amassed an unprecedented level of social and political power. Research has shown that the feedback loops generated by algorithmic recommendations increase polarization~\citep{Jiang2019:degenerate, Krueger2020:hidden,Stoica2019;hegemony}.
Echo chambers created by algorithmic recommendations on these platforms can have a wide range of adverse effects, such as amplifying and creating glass ceilings for minorities~\citep{Stoica2018:algorithmic}, as well as limiting exposure and job recommendations~\cite{Fabbri2022:exposure}. They also lead to disinformation and propaganda  being disproportionately spread to minoritized  groups~\cite{Freelon2022:black}.

In this paper, we propose an approach to content recommendation that simultaneously preserves the positive aspects of personalization while avoiding the pitfalls of filter bubbles.

We do so by introducing a model that ensures that if some users are served a particular category of content, then all users will see at least a small amount of that content.
For example, if a network includes individuals across a political spectrum, then every user will be exposed to at least a small amount of news from opposing perspectives.
This allows a platform to present diverse content without forcing content on its users that no one is interested in.
This approach builds upon seminal work by \citet{Celis19:Controlling} who initiated the study of algorithmic approaches to reducing polarization. However, our approach to avoiding filter bubbles is different and our analysis techniques diverge significantly, as detailed in Section~\ref{sec:related}.

We model a platform recommending content to users with a standard multi-armed bandit formulation.
There are $k$ categories of content---such as fashion, sports, left-learning political content, right-leaning political content, and so on---and $n$ users. For each user and content category, the platform receives a stochastic reward from an unknown distribution for showing the user content from that category, measured, for example, in terms of engagement or ad revenue.
The platform interacts with the $n$ users over $T$ timesteps, at each timestep choosing a distribution over content categories for each user. The platform's goal is to maximize its cumulative reward.
Standard bandit algorithms would eventually learn for each user the category with maximal expected reward and only show them content from that category, at which point the user's content recommendations would be caught in a filter bubble.

\subsection{Our contributions}

We propose a flexible approach to disincentivizing filter bubbles that adapts to the interests of the individuals on the network. We summarize our contributions along the following two axes.

\subsubsection{Modeling contributions}\label{sec:modeling}
We first analyze an approach that requires that the distribution of content shown to any one user is not far from the distribution shown to the population, so users cannot be siloed into disjoint filter bubbles. However, we show that the optimal recommendations exacerbate \emph{tyranny of the majority}: the burden of diversification is borne by groups with minority interests (as often happens with naive approaches to diversification). A majority group will exclusively see content that they most enjoy while recommendations for minority users become far less relevant.

\paragraph{An equitable approach to preventing filter bubbles.} The intuition behind our revised approach is that in order to avoid filter bubbles \emph{and} tyranny of the majority, (1) users should primarily see content that they are most interested in (thus avoiding tyranny of the majority), and (2) if some users are shown a particular type of content, then all users should see at least a small amount of that content (thus avoiding filter bubbles). When both requirements are satisfied, users with majority interests will be exposed to content that interests minority groups and vice versa.

Formally,
for each user, we impose the following constraint: we require that the probability she is shown content from a particular category must be at least $\gamma$ times the average probability the entire population is shown content from that category, where $\gamma \in [0,1]$ is a tunable parameter. We refer to this model as \textbf{Formulation 1.} Setting $\gamma = 0$ corresponds to complete personalization and setting $\gamma = 1$ requires that everyone see the same distribution of content. Moreover, if no one on the network is interested in some type of content, there is no requirement that users be shown that content. When $\gamma \leq \frac{1}{2}$, we show that conditions (1) and (2) are met and thus the burden of diversification is borne more equally among all users.
We also provide a second formulation, called \textbf{Formulation 2}, where instead of imposing hard constraints, the platform is penalized based on the the extent to which it violates the $\gamma$ constraint.

\paragraph{Taxation without knowledge of the true content distribution.} The penalization described above depends on the true, underlying probabilities that the platform assigns to different types of content at each timestep. To augment the flexibility of our approach, we also analyze a model where an auditor only has access to a dataset describing the types of content that users were actually shown, as opposed to a description of the true distributions. In this model, the platform is penalized at the end of the $T$ timesteps based on the extent to which the \emph{empirical} distribution over content shown to each user violates the $\gamma$ constraint described above. We refer to this model as \textbf{Formulation 3.}

\subsubsection{Technical contributions}\label{sec:technical}

Since the platform does not know the reward distributions (corresponding to the users' preferences for the different types of content), it must learn a high-reward policy over the course of the $T$ rounds. We analyze the \emph{regret} of the Upper-Confidence-Bound (UCB) algorithm. The key challenges we face are providing nearly-matching lower bounds---which depend on structure exhibited by the specific constraints that we impose---and bounding the regret under Formulation 3, under which the optimal policy may be \emph{history-dependent.}

\paragraph{Regret upper bounds.} Under Formulation 1, we measure regret as the difference between (1) the cumulative reward of the optimal distribution over content that satisfies our $\gamma$ constraint and (2) the cumulative reward of the platform's learning algorithm. Crucially, the optimal distribution (1) is defined by the users' reward distributions, but these are unknown to the learning algorithm. When $\gamma = 1$, a variant of the UCB algorithm achieves a regret of $\tilde O(\sqrt{nkT})$ and for $\gamma < 1$, another variant achieves a regret of $\tilde O(n\sqrt{kT}).$
Under Formulations 2 and 3, we measure regret with respect to the optimal policy that maximizes the cumulative reward minus the penalty. Our regret bounds are $\tilde O(n\sqrt{kT})$.

\textbf{Key challenge.} Under Formulation 3, the optimal policy may be \emph{history-dependent}: it may dynamically adjust its recommendations based on the empirical distribution over content thus far, and thus the magnitude of the final penalty. This is in contrast to Formulations 1 and 2, where the optimal policy is a fixed distribution over content.

\paragraph{Regret lower bounds.} We provide a nearly-matching lower bound on regret under Formulation 1. As in the upper bound,
our lower bound transitions from an $\Omega(n)$ dependence for small $\gamma$ to an $\Omega(\sqrt{n})$ dependence for large $\gamma$. For $k= 2$ arms, we prove a lower bound of $\Omega(n\sqrt{T})$ for $\gamma < \frac{1}{2}$. Meanwhile, for all $k \geq 2$ and all $\gamma \in [0,1]$, we prove a lower bound of $\Omega(\sqrt{nkT})$.
This means that no algorithm has regret better than $\Omega(n\sqrt{T})$ for $\gamma < \frac{1}{2}$ or $\Omega(\sqrt{nkT})$ for any $\gamma \in [0,1]$.

This transition from a $\Theta(n)$ to $\Theta(\sqrt{n})$ dependence elucidates a tension between the reward of the optimal policy and the ability of the learning algorithm to compete with the optimal policy.
As $\gamma$ grows, the set of distributions that the platform can show the user while still satisfying the $\gamma$ constraint shrinks. Thus, the optimal policy comes from an increasingly restricted set so the regret benchmark is smaller. Likewise, as $\gamma$ grows, the learner has to use an increasingly restricted set of policies to compete with the optimal policy. Since regret shrinks as $\gamma$ grows, we show that 
the optimal policy's reward diminishes at a faster rate than the learner's handicap in competing with the optimal policy.

\textbf{Key challenge.} Lower bounds for bandit problems typically follow by identifying two worst-case problem instances that are similar enough that any algorithm would not be able to statistically distinguish between them, but are distinct enough to ensure that even if an algorithm has low regret on one instance, it will have high regret on the other. Simply creating $n$ copies (one for each user) of the worst-case problem instances used in standard bandit lower bounds would lead to a large statistical difference between problem instances, thus precluding an $\Omega(n)$ dependence. Our lower bound construction therefore takes advantage of structure specific to our model.

\paragraph{Experiments.} We analyze the optimal policies under the formulations from Section~\ref{sec:modeling} using real user preference data~\cite{Harper15:Movielens}. We empirically verify that when users' preferences are heterogeneous, subgroups share the burden of diversification. We also show that users experience only a minor loss in utility when recommended diversified content.

\subsection{Related work}\label{sec:related}
There has been significant interest in understanding the mechanics of how recommender systems affect large-scale opinion dynamics, and if and when they lead to polarization~\citep[e.g.,][]{Haghtalab2021:belief, Bail2018:exposure, Rychwalska2018:polarization}. Most of the analysis has focused on how recommender systems impact network structure \cite{Su2016:effect} and how 
this affects the spread of information and the opinions of members on the network. 
Recently there have been growing calls to algorithmically increase ``exposure diversity'' and combat filter bubbles~\cite{Bozdag2015:breaking, Elahi2022:towards, Helberger2018:exposure}. \citet{Castells2021:novelty} discuss  methodologies and metrics to assess recommendation diversity, and \citet{Halpern2023:optimal} analyze the trade-off between diversity and engagement in recommendation algorithms.

The most related research to ours is seminal work by \citet{Celis19:Controlling}, who initiated the study of algorithmic approaches to reducing polarization. There are a variety of differences between our work and theirs, highlighted below.
\begin{itemize}
    \item \emph{Modeling approach.} \citet{Celis19:Controlling} suggest that a regulator should place pre-determined, fixed upper and lower bounds on the probability that each arm is played so that no user can exclusively see one type of content. Choosing bounds for each type of content, however, may be challenging.
(For example, how should bounds on fashion content and major world events compare?) 
Moreover, if no user is interested in a type of content, it may not make sense to force all users to see it. The regulator would have to make these differential decisions, which would be a divisive and controversial task. These concerns are largely ameliorated under our model.
    \item \emph{Stronger assumption on the regulator's knowledge.} \citet{Celis19:Controlling} assume the regulator can control the exact probabilities that the platform shows different types of content to users. In contrast, in our Formulation 3, we propose a tax based on the content that the platform actually showed the user. As we describe in Section~\ref{sec:technical}, this introduces technical challenges in providing a no-regret algorithm for the platform.
    \item \emph{Lower bounds.} Our nearly-matching lower bounds help develop a complete understanding of this problem.
\end{itemize}

Since the multi-armed bandit problem was proposed~\cite{Thompson1933:likelihood}, many variants 
have been studied, such as bandits with budgets~\cite{Badanidiyuru13:Bandits, Slivkins19:Introduction, Agrawal14:Bandits}, bandits with constraints~\cite{Pacchiano21:Stochastic, Amani19:Linear, Moradipari21:Safe, Dani08:Stochastic}, and bandits with floors on content~\cite{ Wu:16conservative, Claure20:Multi}.
Only a few variants~\cite[e.g.,][]{Hossain2021:fair} study multi-agent settings. However, they usually still involve a common reward like in the classical multi-armed bandit problem. There has also been recent work on fairness in multi-armed bandits~\cite[e.g.,][]{Hossain2021:fair, Joseph2016:fairness} but none of these focus on the issues of filter bubbles and polarization in social networks.

\section{Notation and model}\label{sec:notation}
We use $\cP^{d-1} = \{\vec{x} \in [0,1]^d : \norm{\vec{x}}_1 = 1\}$ to denote the probability simplex and $[k]$ to denote the set $[k] = \{1, 2, \dots, k\}.$

\paragraph{Problem definition.} There are $n$ users and $k$ categories of content---for example, fashion, sports, right-leaning news, left-leaning news, and so on---each modeled as an \emph{arm} of a \emph{$k$-armed bandit}. An instance of our problem, denoted $\nu = \left\{\cD_{i,j} : i \in [n], j \in [k]\right\}$, is defined by reward distributions $\cD_{i,j}$ over $[0,1]$ with density function $f_{i,j}: [0,1] \to \R_{\geq 0}$. This distribution  models the platform's reward for showing user $i$ content from category $j$, measured in terms of engagement or ad revenue, for example. The set of all instances $\nu$ is denoted $\cE^{n,k}$.
The mean of user $i$'s reward distribution for arm $j$ is denoted $\mu_{i,j} \in [0,1]$, with $\vec{\mu}_i = \left(\mu_{i,1}, \dots, \mu_{i,k}\right)$. The instance $\nu$ is unknown to the platform.

\paragraph{Interaction between platform and users.} This interaction takes place over  $T$ timesteps.
At each timestep $t \in [T]$:
\begin{enumerate}
\item The platform selects an \emph{action},
which is a distribution over arms for each user. This distribution corresponds to a random variable $\vec{A}_t \in [k]^n$ over arm choices for each of the $n$ users.
We use the notation $\vec{a}_t \in [k]^n$ to denote the specific set of arms the platform plays on round $t$, so it is a realization of the random variable $\vec{A}_t.$
\item Given the set of arms $\vec{a}_t =  \left(a_{t,1}, \dots, a_{t,n}\right) \in [k]^n$, the platform receives a reward for each user. The reward for user $i$ is drawn from the distribution $\cD_{i, a_{t,i}}$. We use the random variable $\vec{X}_t = \left(X_{t,1}, \dots, X_{t,n}\right) \in [0,1]^n$ to denote the platform's reward on round $t$. We also use $\vec{x}_t \in [0,1]^n$ to denote a realization of this random variable.
\end{enumerate}

\paragraph{Platform's learning algorithm.} The platform uses a learning algorithm, or \emph{policy}, $\pi$ to decide the distribution over arms at each timestep.
On timestep $t \in [T]$, the (randomized) policy $\pi$ takes as input the history $\vec{h}_{t-1} = \left(\vec{a}_1, \vec{x}_1, \dots, \vec{a}_{t-1}, \vec{x}_{t-1}\right) \in ([k]^n \times [0,1]^n)^{t-1}$ and returns the set of arms $\vec{a}_t \in [k]^n$ that will be played on round $t$. The conditional probability that $\vec{A}_t = \vec{a}_t$ given the history $\vec{A}_1 = \vec{a}_1, \vec{X}_1 = \vec{x}_1, \dots, \vec{A}_{t-1} = \vec{a}_{t-1}, \vec{X}_{t-1} = \vec{x}_{t-1}$ is denoted $\pi(\vec{a}_t \mid \vec{a}_1, \vec{x}_1, \dots, \vec{a}_{t-1}, \vec{x}_{t-1})$, or more compactly as $\pi(\vec{a}_t \mid \vec{h}_{t-1})$. The notation $\Pi^{n,k}$ denotes the set of all policies $\pi$.

\paragraph{Distribution over outcomes.} Since the reward distributions are independent, the conditional distribution of the reward $\vec{X}_t \in [0,1]^n$ given $\vec{A}_{t} = \vec{a}_t =  \left(a_{t,1}, \dots, a_{t,n}\right) \in [k]^n$ has density function \[f_{\vec{a}_t}\left(\vec{x}_t\right) = \prod_{i = 1}^n f_{i, a_{t,i}}\left(x_{t,i}\right).\]
The interaction between the policy $\pi$ and the instance $\nu$ induces a distribution $\P_{\pi \nu}$ over outcomes with density function \begin{equation}f_{\pi \nu}\left(\vec{a}_1, \vec{x}_1, \dots, \vec{a}_T, \vec{x}_T\right) = \prod_{t = 1}^T \pi(\vec{a}_t \mid \vec{a}_1, \vec{x}_1, \dots, \vec{a}_{t-1}, \vec{x}_{t-1})f_{\vec{a}_t}(\vec{x}_t).\label{eq:PDF}\end{equation}

\paragraph{Platform's goal.} The platform's overall goal is to choose a policy $\pi$ that optimizes its total reward \begin{equation}\E_{\pi \nu} \left[\sum_{i = 1}^n \sum_{t = 1}^T X_{i,t}\right].\label{eq:exp_reward} \end{equation} For each user $i \in [n]$, the optimal policy would choose the arm $j_i$ that maximizes expected reward: $j_i = \argmax_{j \in [k]}\left\{\mu_{i,j}\right\}$. Classic bandit algorithms will eventually converge to this policy. However, repeatedly showing user $i$ content from category $j_i$ traps the user in a filter bubble. In the next sections, we limit the platform's ability to form filter bubbles.

\section{A first attempt to disincentivize filter bubbles}\label{sec:first}
We begin with a naive first attempt at disincentivizing filter bubbles and show that it has the harsh unintended consequence of exacerbating ``tyranny of the majority'': the burden of diversification is borne by those with minority interests. Interestingly, this issue mirrors real-world attempts at diversification where the labor associated with diversification is put disproportionately on members of the underrepresented groups.

To motivate this first attempt, we observe that in a network with severe filter bubbles, members are partitioned into groups which are exposed to disparate types of content. Thus, our first attempt at avoiding filter bubbles ensures that the content recommendations are not too ``spread out.'' We formalize this intuition by requiring that each user's distribution over content is not too far from the average distribution over content shown to the entire population.

More formally, building on the notation from Section~\ref{sec:notation}, let $\pi_i(j \mid \vec{h}_{t-1})$ denote the marginal probability that the platform shows user $i$ arm $j$ on timestep $t$ given the history $\vec{h}_{t-1}$, with $\vec{\pi}_i(\vec{h}_{t-1}) = \left(\pi_i(1 \mid \vec{h}_{t-1}), \dots, \pi_i(k \mid \vec{h}_{t-1})\right)$. Next, let $\bar{\vec{\pi}}(\vec{h}_{t-1}) = \frac{1}{n}\sum_{i = 1}^n \vec{\pi}_i(\vec{h}_{t-1})$ denote the average of these marginal distributions. The $j^{th}$ component of $\bar{\vec{\pi}}(\vec{h}_{t-1})$, denoted $\bar{\pi}(j \mid \vec{h}_{t-1})$, measures the average probability that arm $j$ is shown to any user.
Under our naive first approach, we require that the distance between the vectors $\vec{\pi}_i(\vec{h}_{t-1})$ and $\bar{\vec{\pi}}(\vec{h}_{t-1})$ is small under the $\ell_{\infty}$-norm: \begin{equation}\norm{\vec{\pi}_i(\vec{h}_{t-1}) - \bar{\vec{\pi}}(\vec{h}_{t-1})}_{\infty} = \max_{j \in [k]}\left|\pi_i(j \mid \vec{h}_{t-1}) - \bar{\pi}(j \mid \vec{h}_{t-1})\right| \leq \Delta\label{eq:delta_constriants}\end{equation} for some $\Delta > 0$. (The $\ell_{\infty}$-norm could be replaced by any norm, but we use the $\ell_{\infty}$-norm for this exposition.)

We now show that the optimal policy $\vec{p}_1^{*}, \dots, \vec{p}_n^{*} \in \cP^{k-1}$ leads to tyranny of the majority, where \[\vec{p}_1^{*}, \dots, \vec{p}_n^{*} = \argmax_{\vec{p}_1, \dots, \vec{p}_n}\left\{\sum_{i = 1}^n \vec{\mu}_i \cdot \vec{p}_i : \norm{\vec{p}_i - \frac{1}{n}\sum_{i' = 1}^n \vec{p}_{i'}}_{\infty} \leq \Delta, \forall i \in [n]\right\}.\]
To illustrate the pitfalls of this approach, we analyze a setting where there are two types of content (e.g., left- and right-leaning political content) and the users can be partitioned into disjoint sets where one set only likes content from the first category (i.e., $\vec{\mu}_i = (1,0)$). Meanwhile, the other set only likes content from the second category (i.e., $\vec{\mu}_i = (0,1)$). Without loss of generality, we assume that the former set---which we denote as $N$---is the majority.

When $\Delta \geq \frac{|N|}{n}$, the constraints are meaningless and allow for full personalization: $\vec{p}_i^* = (1,0)$ if $i \in [N]$ and $\vec{p}_i^* = (0,1)$ if $i \not\in [N]$.
Therefore, we analyze the case where $\Delta < \frac{|N|}{n}$. We show that under the optimal policy, the majority group will be able to exclusively see the content that they enjoy: $\vec{p}_i^* = (1,0)$ if $i \in N$. Meanwhile, the minority group's recommendations take a hit in order to ensure that the constraints are satisfied. In particular, for all $i \not\in N$, $\vec{p}_i^* =\left(1 - \frac{n\Delta}{|N|}, \frac{n\Delta}{|N|}\right)$. The proof of the following lemma is in Appendix~\ref{app:first}.

\begin{restatable}{lemma}{naiveOpt}\label{lem:naive_opt}
Suppose that there are $k = 2$ arms and for some set $N \subseteq [n]$ with $|N| \geq \frac{n}{2}$, $\vec{\mu}_i = (1,0)$ for all $i \in N$ and $\vec{\mu}_i = (0,1)$ for all $i \not\in N$. If $\Delta < \frac{|N|}{n}$, then $\vec{p}_i^* = (1,0)$ if $i \in N$ and  $\vec{p}_i^* =\left(1 - \frac{n\Delta}{|N|}, \frac{n\Delta}{|N|}\right)$ otherwise.
\end{restatable}

Lemma~\ref{lem:naive_opt} illustrates that under this approach, tyranny of the majority prevails at the expense of minority interests.

\section{Equitable approaches to disincentivizing filter bubbles}\label{sec:equitable}
Motivated by Section~\ref{sec:first}, we propose three different formulations for disincentivizing filter bubbles that avoid tyranny of the majority. The intuition behind these approaches is built upon the following two pillars:
\begin{enumerate}
\item To avoid tyranny of the majority, users should primarily be recommended content they are most interested in,
\item But to avoid filter bubbles, that content must contain a flavor of the content shown to the entire population.
\end{enumerate} 
We show that it is possible to achieve both of these ends.
If both conditions are satisfied, then a policy like that of Lemma~\ref{lem:naive_opt} where the majority group sees no minority content is not possible. By the first requirement, groups with minority interests will be recommended content that they are interested in, which means that by the second requirement, the majority group's content recommendations will contain a small amount of that minority content, and vice versa.

\subsection{Formulation 1: Personalization constraint}\label{sec:constraint}
In our first formulation, we require that for each user $i\in [n]$, $\vec{\pi}_i(\vec{h}_{t-1})$ is at least $\gamma \bar{\vec{\pi}}(\vec{h}_{t-1})$ for some $\gamma \in [0,1]$: \begin{equation}\vec{\pi}_i(\vec{h}_{t-1}) \geq \gamma \bar{\vec{\pi}}(\vec{h}_{t-1}).\label{eq:gamma_constraint}\end{equation} Each user's recommendations become less personalized as $\gamma$ grows.

To illustrate the benefit of this approach over that of Section~\ref{sec:first}, we analyze the same polarized example where there is a majority group $N$ with $\vec{\mu}_i = (1,0)$ for all $i \in N$. For the minority group, $\vec{\mu}_i = (0,1)$ for all $i \not\in N.$ For all $\gamma \leq \frac{1}{2}$, we show that under the optimal policy, the majority of each group's content recommendations match their interests, but both groups see some content that appeals to the opposing group. In this case the optimal policy is defined as \begin{equation}\vec{p}_1^{*}, \dots, \vec{p}_n^{*} = \argmax_{\vec{p}_1, \dots, \vec{p}_n}\left\{\sum_{i = 1}^n \vec{\mu}_i \cdot \vec{p}_i : \vec{p}_i \geq \frac{\gamma}{n}\sum_{i' = 1}^n \vec{p}_{i'}, \forall i \in [n]\right\}.\label{eq:opt}\end{equation}
The proof of the following lemma is in Appendix~\ref{app:equitable}.

\begin{restatable}{lemma}{interpretation}\label{lem:interpretation}
    Suppose that there are $k = 2$ arms and for some set $N \subseteq [n]$, $\vec{\mu}_i = (1,0)$ for all $i \in N$ and $\vec{\mu}_i = (0,1)$ for all $i \not\in N$. For $\gamma \leq \frac{1}{2}$, the optimal policy has the following form:
    \[\vec{p}_i^{*} = \begin{cases} \left(1 - \frac{\gamma(n-|N|)}{n}, \frac{\gamma(n-|N|)}{n}\right) & \text{if } i \in N\\
    \left(\frac{\gamma |N|}{n}, 1 - \frac{\gamma |N|}{n} \right) & \text{if } i \not\in N.\end{cases}\]
\end{restatable}

Since $\gamma \leq \frac{1}{2}$, this policy ensures that users are mostly recommended content that they are interested in: $\vec{\mu}_i \cdot \vec{p}_i^{*} \geq 1 - \gamma \geq \frac{1}{2}$ for all $i \in [n].$
However, they are still shown a small fraction of content that the other set of the population is interested in. We note that when $N$ is the majority group $\left(|N| \geq \frac{n}{2}\right)$, the minority group $[n] \setminus N$ still sees more content that they are not interested in than the majority group because $\frac{\gamma|N|}{n} \geq \frac{\gamma(n-|N|)}{n}$. However, the burden of diversification is split far more equally among the two groups than in Lemma~\ref{lem:naive_opt}.
The policy mirrors a typical mode of community forum discussions where members split time between listening to the opinions of each person in the entire group (for a $\gamma$-fraction of the time) and breaking into focus groups about specific topics (for a $(1-\gamma)$-fraction of the time).

In Section~\ref{sec:regret}, we provide upper and lower bounds on the platform's \emph{regret} with respect to the optimal policies $\vec{p}_1^*, \dots, \vec{p}_n^*$ defined in Equation~\eqref{eq:opt}. Regret measures the difference between the total reward of the optimal policy and that of the platform's policy $\pi$.
In other words, for any instance $\nu$ and policy $\pi$, the expected regret is defined as  \begin{equation}R_{T,1}(\pi, \nu) = T\sum_{i = 1}^n \vec{p}_i^{*}\cdot \vec{\mu}_i - \E_{\pi\nu}\left[\sum_{i = 1}^n \sum_{t = 1}^T X_{i,t}\right].\label{eq:regret_def}\end{equation}

\subsection{Formulation 2: Personalization penalty}\label{sec:taxation}
We analyze a second formulation where there are no constraints on the platform's policy, but the platform is penalized based on the extent to  which Equation~\eqref{eq:gamma_constraint} is violated. Given a parameter $\eta \geq 0$, this penalty is defined as \[\eta \sum_{i = 1}^n \sum_{j = 1}^k \max\left\{\gamma \bar{\vec{\pi}}(j \mid \vec{h}_{t-1}) - \pi_i(j \mid \vec{h}_{t-1}), 0\right\}.\]
In other words, the platform's goal is to maximize its cumulative reward \begin{align}&\rew_2(\pi, \nu; \eta, \gamma)\nonumber\\
=\, &\E_{\pi\nu}\left[\sum_{i = 1}^n \left(\sum_{t = 1}^T X_{i,t} - \eta \sum_{j = 1}^k \max\left\{\gamma \bar{\vec{\pi}}(j \mid \vec{h}_{t-1}) - \pi_i(j \mid \vec{h}_{t-1}), 0\right\}\right)\right]\nonumber\\
=\, &\sum_{t = 1}^T \E_{\pi\nu}\left[\sum_{i = 1}^n \left(\vec{\mu}_i \cdot \vec{\pi}_i(\vec{h}_{t-1}) - \eta \sum_{j = 1}^k \max\left\{\frac{\gamma}{n} \sum_{i' =1}^n\pi_{i'}(j \mid \vec{h}_{t-1}) - \pi_i(j \mid \vec{h}_{t-1}), 0\right\}\right)\right].\label{eq:rew2}\end{align}
The policy that maximizes Equation~\eqref{eq:rew2} is history independent and can be written as $\vec{p}^* = \left(\vec{p}_1^*, \dots, \vec{p}_n^*\right)$ with $\vec{p}_i^* \in \cP^{k-1}.$ 
The expected regret of a policy $\pi$ under this formulation is $R_{T,2}(\pi, \nu) = \rew_2(\vec{p}^*, \nu; \eta, \gamma) - \rew_2(\pi, \nu; \eta, \gamma)$.

\subsection{Formulation 3: Personalization penalty on the empirical distribution}

Sections~\ref{sec:constraint} and \ref{sec:taxation} describe models in which the platform is subject to constraints or penalties based on the \emph{true} distribution over content that it shows users. However, an auditor may only have access to the \emph{realizations} of those distributions---that is, the set of arms $a_{t,i}\in [k]$ shown to each user $i$ at timestep $t$. Formulation 3 covers a setting in which a regulator penalizes the platform at the end of the $T$ timesteps based on the empirical distribution over content.
Specifically, let $\hat{p}_{i,j} = \frac{1}{T} \sum_{t = 1}^T \textbf{1}_{\{A_{t,i} = j\}}$ be the average number of times that the platform pulls arm $j$ for user $i$. At the end of the $T$ timesteps, the platform is penalized based on how small $\hat{p}_{i,j}$ is compared to  $\frac{\gamma}{n} \sum_{i'=1}^n \hat{p}_{i',j}$. In particular, given a normalizing factor $\eta$, we define a penalty that is the analogue of Equation~\eqref{eq:rew2}: \[\eta \sum_{i = 1}^n \sum_{j= 1}^k \max\left\{\frac{\gamma}{n} \sum_{i'=1}^n \hat{p}_{i',j} - \hat{p}_{i,j}, 0\right\}.\]

The platform's goal is therefore to maximize their expected total payoff minus this penalty, which is equal to \begin{align}\rew_3(\pi, \nu; \eta, \gamma) &= \E_{\pi\nu}\left[\sum_{i = 1}^n\left(\sum_{t = 1}^T X_{i,t} - \eta \sum_{j= 1}^k \max\left\{\frac{\gamma}{n} \sum_{i'=1}^n \hat{p}_{i',j} - \hat{p}_{i,j}, 0\right\}\right)\right]\nonumber\\
&= \sum_{i = 1}^n\left(\sum_{t = 1}^T \E_{\pi\nu}\left[\vec{\mu}_i \cdot \vec{\pi}_i(\vec{h}_{t-1})\right] - \eta \sum_{j= 1}^k \E_{\pi\nu}\left[\max\left\{\frac{\gamma}{n} \sum_{i'=1}^n \hat{p}_{i',j} - \hat{p}_{i,j}, 0\right\}\right]\right).\label{eq:rew3}\end{align}
Let $\pi^*$ be the policy that maximizes Equation~\eqref{eq:rew3}. The regret of $\pi$ is \[R_{T,3}(\pi, \nu) = \rew_3(\pi^*, \nu; \eta, \gamma) - \rew_3(\pi, \nu; \eta, \gamma).\]

A key difference between Equation~\eqref{eq:rew2} and Equation~\eqref{eq:rew3} is that in Equation~\eqref{eq:rew2}, the platform is penalized at every timestep whereas in Equation~\eqref{eq:rew3}, the platform is penalized at the end of the $T$ timesteps. We make this distinction because the empirical distribution over content at a single timestep would be extremely noisy.

\section{Regret analysis}\label{sec:regret}
In this section, we discuss algorithms that the platform can use to minimize regret in the three formulations from Section~\ref{sec:equitable}. We also provide a nearly-matching lower bound on regret for Formulation 1 in Section~\ref{sec:lower}.

\subsection{Regret analysis for Formulation 1}
We begin with lower bounds on regret under Formulation 1.
In Section~\ref{sec:regret_small_gamma}, we show that a variant of the UCB algorithm has regret $O(n\sqrt{Tk})$ for $\gamma < 1$ and in Section~\ref{sec:regret_gamma_one}, we show that a different variant of UCB has regret $O(\sqrt{nkT})$ for $\gamma = 1$. We then prove in Section~\ref{sec:lower} that these bounds are nearly optimal: for $\gamma \leq \frac{1}{2}$ and $k = 2$, no algorithm can achieve regret better than $\Omega(n\sqrt{T})$, and for all $k \geq 2$ and $\gamma \in [0,1]$ (including $\gamma > \frac{1}{2}$) our bound is $\Omega(\sqrt{nkT})$.

The transition from a $\Theta(n)$ to a $\Theta(\sqrt{n})$ dependence illustrates that as $\gamma$ grows, the platform is better able to compete with the optimal policy subject to the $\gamma$ constraints. As $\gamma$ grows, the platform has a smaller set of distributions that it can use to compete with the optimal policy while obeying the $\gamma$ constraints. However, for the same reason, the cumulative reward of the optimal policy shrinks as $\gamma$ grows. Intuitively, the transition from a $\Theta(n)$ to a $\Theta(\sqrt{n})$ dependence as $\gamma$ grows illustrates that the optimal policy's reward degrades faster than the platform's ability to compete with that policy.

\subsubsection{Regret upper bound when $\gamma < 1$}\label{sec:regret_small_gamma}

We analyze a multi-agent variant of the UCB algorithm, which we call \textsc{$n$-UCB}, and show that it has regret $O(n\sqrt{Tk})$ when $\gamma < 1$.

The $n$-UCB algorithm essentially runs a copy of classic UCB for each user, but coordinates amongst these $n$ UCB copies to ensure that they satisfy the global constraints. This requires $n$-UCB to play distributions over arms from the set of distributions $\left(\vec{p}_1, \dots, \vec{p}_n\right)$ that satisfy the constraints: $\vec{p}_i \geq  \frac{\gamma}{n}\sum_{i' = 1}^n \vec{p}_{i'}$ for all $i \in [n]$. This is in contrast to the classic case where UCB plays a single arm at each timestep.
For completeness, we include a full description of $n$-UCB (Algorithm~\ref{alg:ucb.f2}) and the proof of the following theorem in Appendix~\ref{app:regret_small_gamma}.

\begin{restatable}{theorem}{nUCB}\label{thm:low_regret.f2}
    Let $\pi$ be the policy of $n$-UCB. Then $R_{T,1}(\pi, \nu) = \tilde O( n\sqrt{kT})$. 
\end{restatable}

\subsubsection{Regret upper bound when $\gamma = 1$}\label{sec:regret_gamma_one}

When $\gamma = 1$, all users must be shown the same distribution of content. We can therefore reduce our problem to a single-agent bandit problem with the reward distributions $\cD_j = \sum_{i=1}^n \cD_{i,j}$ for all arms $j \in [k]$. We adapt the robust-UCB framework by~\citet{Bubeck13:Bandits} with the median-of-means estimator~\citep{Alon1996:moments}, as summarized by Algorithm~\ref{alg:robustucb.f2} in Appendix~\ref{app:regret_gamma_one}. The full proof of the following theorem is in Appendix~\ref{app:regret_gamma_one}.

\begin{restatable}{theorem}{regretGammaOne}\label{thm:regretGammaone}
    Let $\pi$ be the policy of Robust-UCB. Then $R_{T,1}(\pi, \nu) = \tilde O (\sqrt{nkT}).$
\end{restatable}

\subsubsection{Regret lower bound}\label{sec:lower}
In this section, we provide nearly-matching regret lower bounds. Our first bound holds when there are $k = 2$ arms, $\gamma \leq \frac{1}{2}$, and $n$ is sufficiently large ($n > 100$). In this case, we prove a regret lower bound of $\Omega(n\sqrt{T})$. Meanwhile, for all $k \geq 2$ and $\gamma \in [0,1]$ (including $\gamma > \frac{1}{2}$), we provide a bound of $\Omega(\sqrt{nkT})$.

We begin with our main result (Theorem~\ref{thm:lower}) and show in Corollary~\ref{cor:large_n} 

\begin{theorem}\label{thm:lower}
   For all  $T \geq 4$, the regret is lower bounded as follows: \[\inf_{\pi \in \Pi^{n,2}}\sup_{\nu \in \cE^{n,2}} R_{T,1}(\pi, \nu) \geq \max\left\{\sqrt{\frac{T}{8}}\left(\frac{n}{8e} - \gamma \left(\frac{n}{8e} + \sqrt{\frac{n}{2\pi}}\right) \right), \frac{\sqrt{nT}}{16e}\right\}.\]
\end{theorem}

\begin{proof}
This theorem follows directly from Lemmas~\ref{lem:nsqrtT} and \ref{lem:sqrtnT}.
\end{proof}

\begin{corollary}\label{cor:large_n}
    For
all $n > 100$, $\gamma \leq \frac{1}{2}$, and $T \geq 4$, the regret is lower bounded as \[\inf_{\pi \in \Pi^{n,2}}\sup_{\nu \in \cE^{n,2}} R_{T,1}(\pi, \nu) \geq \frac{n\sqrt{T}}{900}.\]
\end{corollary}

We now provide a proof sketch of the first part of Theorem~\ref{thm:lower}. The full proof is in Appendix~\ref{app:lower_bound}.

\begin{restatable}{lemma}{nsqrtT}\label{lem:nsqrtT}
For all $T \geq 1$, the regret is lower bounded as follows: \begin{equation}\inf_{\pi \in \Pi^{n,2}}\sup_{\nu \in \cE^{n,2}} R_{T,1}(\pi, \nu) \geq \sqrt{\frac{T}{8}}\left(\frac{n}{8e} - \gamma \left(\frac{n}{8e} + \sqrt{\frac{n}{2\pi}}\right) \right).\label{eq:1regret_lb}\end{equation}
\end{restatable}

\begin{proof}[Proof sketch]
Our proof is based on worst-case instances $\nu_{\vec{b}}$ defined for any vector $\vec{b} \in \{0,1\}^n$.
For each agent $i \in [n]$, their reward distributions for the two arms are Bernoulli with means $\vec{\mu}_i = \left(\mu_{i,0}, \mu_{i,1}\right)$ where $\vec{\mu}_i = \left(\frac{1}{2} + \epsilon, \frac{1}{2}\right)$ if $b_i = 0$, $\vec{\mu}_i = \left(\frac{1}{2}, \frac{1}{2} + \epsilon\right)$ if $b_i = 1$, and $\epsilon = \sqrt{\frac{1}{8T}}.$
We lower bound the expected regret $\E\left[R_{T,1}\left(\pi, \nu_{\vec{b}}\right)\right]$ by the right-hand-side of Equation~\eqref{eq:1regret_lb}, where the expectation is over both the draw of the vector $\vec{b} \sim \text{Unif}\left(\{0,1\}^n\right)$ and the distribution over outcomes $\P_{\pi\nu_{\vec{b}}}$. This implies that there exists an instance $\nu_{\vec{b}}$ such that $R_{T,1}\left(\pi, \nu_{\vec{b}}\right) \geq \sqrt{\frac{T}{8}}\left(\frac{n}{8e} - \gamma \left(\frac{n}{8e} + \sqrt{\frac{n}{2\pi}}\right) \right)$.

Without constraints, the optimal policy would exclusively show arm $0$ to each user $i \in [n]$ with $b_i = 0$ since it has higher reward for these users, and similarly it would exclusively show arm $1$ to each user $i \in [n]$ with $b_i = 1$. Due to the constraints, both the optimal policy and the policy $\pi$ must show some users the ``wrong'' arm on a non-negligible fraction of rounds. In total, the policy $\pi$ will lose the following reward from showing users the wrong arms:
\[\epsilon \E_{\vec{b}}\left[\sum_{t = 1}^T\left(\sum_{i : b_i = 0}\E_{\pi\nu_{\vec{b}}}\left[\pi_i\left(1 \mid \vec{h}_{t-1}\right)\right] + \sum_{i : b_i = 1} \E_{\pi\nu_{\vec{b}}}\left[\pi_i\left(0 \mid \vec{h}_{t-1}\right)\right]\right)\right].\]

We begin by proving that that optimal policy loses at most \begin{equation}\sqrt{\frac{T}{8}}\cdot \frac{\gamma(n-1)}{2}\label{eq:opt_policy_loss}\end{equation} total reward from showing users the wrong arms. Meanwhile, we prove that any policy $\pi$ will lose at least \begin{equation}\sqrt{\frac{T}{8}}\left(\frac{\gamma n}{2} + \frac{n}{8e} - \gamma \left(\frac{n}{8e} + \sqrt{\frac{n}{2\pi}}\right) \right)\label{eq:learned_policy_loss}\end{equation} total reward.
We prove this by showing that for any user $i \in [n]$, the distribution over outcomes conditioned on $b_i = 0$ is close to the distribution over outcomes conditioned on $b_i = 1$. Intuitively, this means that any policy $\pi$ will struggle to distinguish whether $b_i = 0$ or $b_i = 1$.
Taking the difference of \eqref{eq:learned_policy_loss} and \eqref{eq:opt_policy_loss} implies the lemma. 
\end{proof}

We conclude by proving that for all $\gamma \in [0,1]$ and $k \geq 2$, regret is lower bounded by \[\frac{1}{16e}\sqrt{nT(k-1)}.\] The proof is similar to that of existing bandit lower bounds~\citep[e.g.,][Theorem 15.2]{Lattimore20:Bandit}, so we include it for completeness in Appendix~\ref{app:lower_bound}.

\begin{restatable}{lemma}{sqrtnT}\label{lem:sqrtnT}
For all $T \geq \frac{7(k-1)}{n}$, the regret is lower bounded as follows: \[\inf_{\pi \in \Pi^{n,k}}\sup_{\nu \in \cE^{n,k}} R_{T,1}(\pi, \nu) \geq \frac{\sqrt{nT(k-1)}}{16e}.\]
\end{restatable}

\subsection{Regret analysis for Formulation 2}\label{sec:2_upper_bound}
Under Formulation 2, a variation on UCB we call Penalty-UCB (Algorithm~\ref{alg:ucb.f1}) achieves regret $\tilde O(n\sqrt{kT})$. Penalty-UCB maintains estimates $\vec{\hat{\mu}}_i^{(t)}$ of each $\vec{\mu}_i$ and selects the distribution maximizing the estimated reward minus the penalty: \[\left(\vecp_i^{(t)}\right)_{i \in [n]} = \amax\left\{\sum\limits_{i = 1}^n \vecp_i \cdot \vec{\hat{\mu}}_i^{(t)}- \eta \sum_{j = 1}^k \max\left\{\frac{\gamma}{n} \sum_{i' =1}^np_{i',j} - p_{i,j}, 0\right\}\right\}.\] For completeness, we include the proof of the following theorem in Appendix~\ref{app:2_regret}.

\begin{restatable}{theorem}{formtwo}\label{thm:taxation}
Let $\pi$ be the policy of Penalty-UCB. Then $R_{T,2}(\pi, \nu) = \tilde O(n\sqrt{kT}).$
\end{restatable}

\subsection{Regret analysis for Formulation 3}
A key challenge under Formulation 3 is that the platform's optimal strategy, given perfect information about the reward distributions $\cD_{i,j}$, may be \emph{history dependent}. For example, the platform may choose to increase or decrease personalization dynamically based on the empirical distribution of content chosen thus far.

Nonetheless, we show that Penalty-UCB (Algorithm~\ref{alg:ucb.f1}) competes with the optimal history-dependent policy by reducing our analysis to that of Section~\ref{sec:2_upper_bound}.
 We use the notation $\pi^*$ to denote the optimal policy that maximizes Equation~\eqref{eq:rew3}.

First, we show that under Formulation 2, the optimal policy obtains a larger reward (measured in terms of $\rew_2$) than $\pi^*$ under Formulation 3 (measured in terms of $\rew_3$). The full proof is in Appendix~\ref{app:3_regret}.

\begin{restatable}{lemma}{ubOptReward}\label{lem:ub_opt_reward}
Let $\vec{p}^* = \left(\vec{p}_1^*, \dots, \vec{p}_n^*\right)$ with $\vec{p}_i^* \in \cP^{k-1}$ be the policy that maximizes $\rew_2\left(\vec{p}, \nu; \frac{\eta}{T}, \gamma\right)$.  Then \[\rew_2\left(\vec{p}^*, \nu; \frac{\eta}{T}, \gamma\right) \geq \rew_3(\pi^*, \nu; \eta, \gamma).\]

\end{restatable}

\begin{proof}[Proof sketch]
For any arm $j\in [k]$, we can exchange the expectation and the maximum in Equation~\eqref{eq:rew3} as follows: \[\E_{\pi^*\nu}\left[\max\left\{\frac{\gamma}{n} \sum_{i'=1}^n \hat{p}_{i',j} - \hat{p}_{i,j}, 0\right\}\right] \geq \max\left\{\frac{\gamma}{n}\sum_{i'=1}^n  \E_{\pi^*\nu}\left[\hat{p}_{i',j}\right] -  \E_{\pi^*\nu}\left[\hat{p}_{i,j}\right], 0\right\}.\]
By definition of $\E_{\pi^*\nu}\left[\hat{p}_{i,j}\right]$, this allows us to show that $\rew_3(\pi^*, \nu; \eta, \gamma)$ is upper-bounded by
\begin{align}\sum_{i = 1}^n\sum_{j = 1}^k\nonumber &\left(\sum_{t = 1}^T \mu_{i,j}\E_{\pi^*\nu}\left[\pi^*_i(j \mid \vec{h}_{t-1})\right]\right.\\
&\left. - \eta \max\left\{\frac{1}{T} \sum_{t = 1}^T\left(\frac{\gamma}{n}\sum_{i'=1}^n \E_{\pi^*\nu}\left[\pi^*_{i'}(j \mid \vec{h}_{t-1})\right] - \E_{\pi^*\nu}\left[\pi^*_i(j \mid \vec{h}_{t-1})\right]\right), 0\right\}\right).\label{eq:exp_inside_main}
\end{align}
We next define the history-independent policy $\vec{p} = \left(\vec{p}_1, \dots, \vec{p}_n\right)$ such that \[p_{i,j} = \frac{1}{T}\sum_{t = 1}^T \E_{\pi^*\nu}\left[\pi_i^*(j \mid \vec{h}_{t-1})\right].\] We rearrange Equation~\eqref{eq:exp_inside_main} and use the definition of $\vec{p}^*$ to get that \begin{align*}\rew_3(\pi^*, \nu; \eta, \gamma) &\leq \sum_{i = 1}^n\sum_{j = 1}^k\left(T \mu_{i,j}p_{i,j} - \eta \max\left\{\frac{\gamma}{n}\sum_{i'=1}^n p_{i', j} - p_{i,j}, 0\right\}\right)\\
&\leq \sum_{i = 1}^n\sum_{j = 1}^k\left(T \mu_{i,j}p^*_{i,j} - \eta \max\left\{\frac{\gamma}{n}\sum_{i'=1}^n p^*_{i', j} - p^*_{i,j}, 0\right\}\right)= \rew_2\left(\vec{p}^*, \nu; \frac{\eta}{T}, \gamma\right).\end{align*}
\end{proof}

Next, we show that for any policy $\pi$ that deterministically plays each of the $k$ arms once in the first $k$ rounds, the difference between its rewards under Formulations 2 and 3 is bounded. This condition holds for Penalty-UCB (Algorithm~\ref{alg:ucb.f1}) and could be removed with a slightly more involved analysis. The full proof is in Appendix~\ref{app:3_regret}.

\begin{restatable}{lemma}{lbAnyPolicy}\label{lem:lb_any_policy}
Let $\pi$ be any policy such that $\pi_i(t \mid \vec{h}_{t-1}) = 1$ for all $t \leq k$ and $i \in [n]$.
    For any instance $\nu$,
    \[\rew_2\left(\pi, \nu; \frac{\eta}{T}, \gamma\right) \leq \rew_3(\pi, \nu; \eta, \gamma)+ \eta nk(\gamma+1)\sqrt{\frac{10\log T}{T}}.\]
\end{restatable}

\begin{proof}[Proof sketch]
    For any arm $j \in [k]$, we can exchange the expectation and the maximum in Equation~\eqref{eq:rew2} as follows: \begin{align}&\frac{1}{T}\sum_{t = 1}^T\E_{\pi\nu}\left[\max\left\{\frac{\gamma}{n} \sum_{i'=1}^n \pi_{i'}(j \mid \vec{h}_{t-1}) - \pi_{i}(j \mid \vec{h}_{t-1}) , 0\right\}\right] \nonumber\\ \geq \, &\max\left\{\E_{\pi\nu}\left[\frac{1}{T}\sum_{t = 1}^T\left(\frac{\gamma}{n}\sum_{i'=1}^n  \textbf{1}_{\{A_{t,i'} = j\}} -  \textbf{1}_{\{A_{t,i} = j\}}\right)\right], 0\right\}.\label{eq:pi_inside_main}\end{align}
Next, we use a result by \citet{Aven85:Upper} to show that the right-hand-side of Equation~\eqref{eq:pi_inside_main} is lower-bounded by \begin{equation}\E_{\pi\nu}\left[\max\left\{\frac{\gamma}{n} \sum_{i'=1}^n \hat{p}_{i',j} - \hat{p}_{i,j}, 0\right\}\right] - \sqrt{\frac{1}{2T^2}\cdot \text{Var}\left( \sum_{t = 1}^T\left(\frac{\gamma}{n}\sum_{i'=1}^n  \textbf{1}_{\{A_{t,i'} = j\}} -  \textbf{1}_{\{A_{t,i} = j\}}\right)\right)}.\label{eq:aven_main}\end{equation}
We upper-bound the variance term in Equation~\eqref{eq:aven_main} by $20 T (\gamma+1)^2\log T$, which implies the lemma statement.
\end{proof}

Our regret bound follows from Lemmas~\ref{lem:ub_opt_reward} and \ref{lem:lb_any_policy} as well as Theorem~\ref{thm:taxation}. The proof is in Appendix~\ref{app:3_regret}.

\begin{restatable}{theorem}{empRegret}\label{thm:empirical_regret}
Let $\pi$ be the policy played by Algorithm~\ref{alg:ucb.f1}. Then the regret is bounded as 
    \[\rew_3(\pi^*, \nu; \eta, \gamma) - \rew_3(\pi, \nu; \eta, \gamma) = \tilde O\left( n\sqrt{kT} + \frac{\eta nk(1+\gamma)}{\sqrt{T}}\right).\]
\end{restatable}

Even if $\eta$ grows linearly in $T$, the regret bound in Theorem~\ref{thm:empirical_regret} will only grow with $\sqrt{T}$.

\section{Empirical Results}
To explore how our framework impacts exposure diversity in practice, we test it out on real world data: the MovieLens dataset \cite{Harper15:Movielens} which describes people’s expressed preferences for movies\footnote{We use this dataset in order to analyze our methods on real-world user preferences, recognizing that movie recommendation filter bubbles would likely not be as pernicious as political news filter bubbles, for example.}. These preferences take the form of <user, item, rating, timestamp> tuples, each the result of a user giving a 0–5 star rating for a movie at a particular time.

\subsection{Experimental setup}\label{sec:exp_setup}
There are $n = 58$ users, randomly selected from the database,
and a set $\cK$ of $k = 18$ movie genres: $\cK = \{$Action, Adventure, Animation, Children, Comedy, Crime, Documentary, Drama, Fantasy, Film-Noir, Horror, Musical, Mystery, Romance, Sci-Fi, Thriller, War, Western$\}$.
Each genre is paired with an associated index in  $[k]$ determined by alphabetically ordering $\cK$.

For each movie $m \in \cM$ , where $\cM$ is the set of all movies, there is an associated genre set $m_K \subseteq [k]$ with $|m_K| \geq 1$ (a movie could belong to multiple genres). We use the ratings data to generate preferences for the users. For each movie $m \in \cM_i$, where $\cM_i$ is the set of movies watched by user $i \in [n]$, the user gives a numeric rating $r_{i,m}$ on a 5-star scale with half-star increments: $r_{i,m} \in \{0.5, 1, 1.5, \ldots, 5.0\}.$ We sum these ratings by genre and divide by the number of movies 
that user $i$ watched from that genre. This results in an average  rating $\mu_{i,j} \in [0,5]$ per genre $j \in [k]$. Finally, we divide  $\mu_{i,j}$ by 5 so that $\mu_{i,j} \in [0,1]$. In the end, \[\mu_{i,j} = \frac{\sum_{m \in \cM_i} r_{i,m}\cdot\ind{j \in m_K}}{\sum_{m \in \cM_i} \ind{j \in m_K}} \cdot \frac{1}{5}.\] Using the $\vec{\mu}_i$s as the mean reward vectors, we use linear programs (LPs) to solve for the optimal policy under no constraints and under both our polarization cap and polarization tax frameworks.

\subsection{Effect of the polarization cap and tax on content recommendations}\label{sec:exp_prob}

We begin by investigating the effects that our constraints from Formulation 1 (Section~\ref{sec:constraint}) have on the optimal content distribution.

These experiments provide a parallel to Lemma~\ref{lem:interpretation}, which shows that in a polarized population, users share the burden of diversification. To model a polarized society, we restrict our attention to two dissimilar genres: thriller and romance. In this restricted space, $\vec{\mu}_i \in [0,1]^2$.
We call the users who prefer the thriller genre \emph{thriller-lovers} and those who prefer the romance genre \emph{romance-lovers}.

\begin{figure}
    \centering
    \begin{subfigure}[b]{0.48\textwidth}
        \centering
        \includegraphics[width=\textwidth]{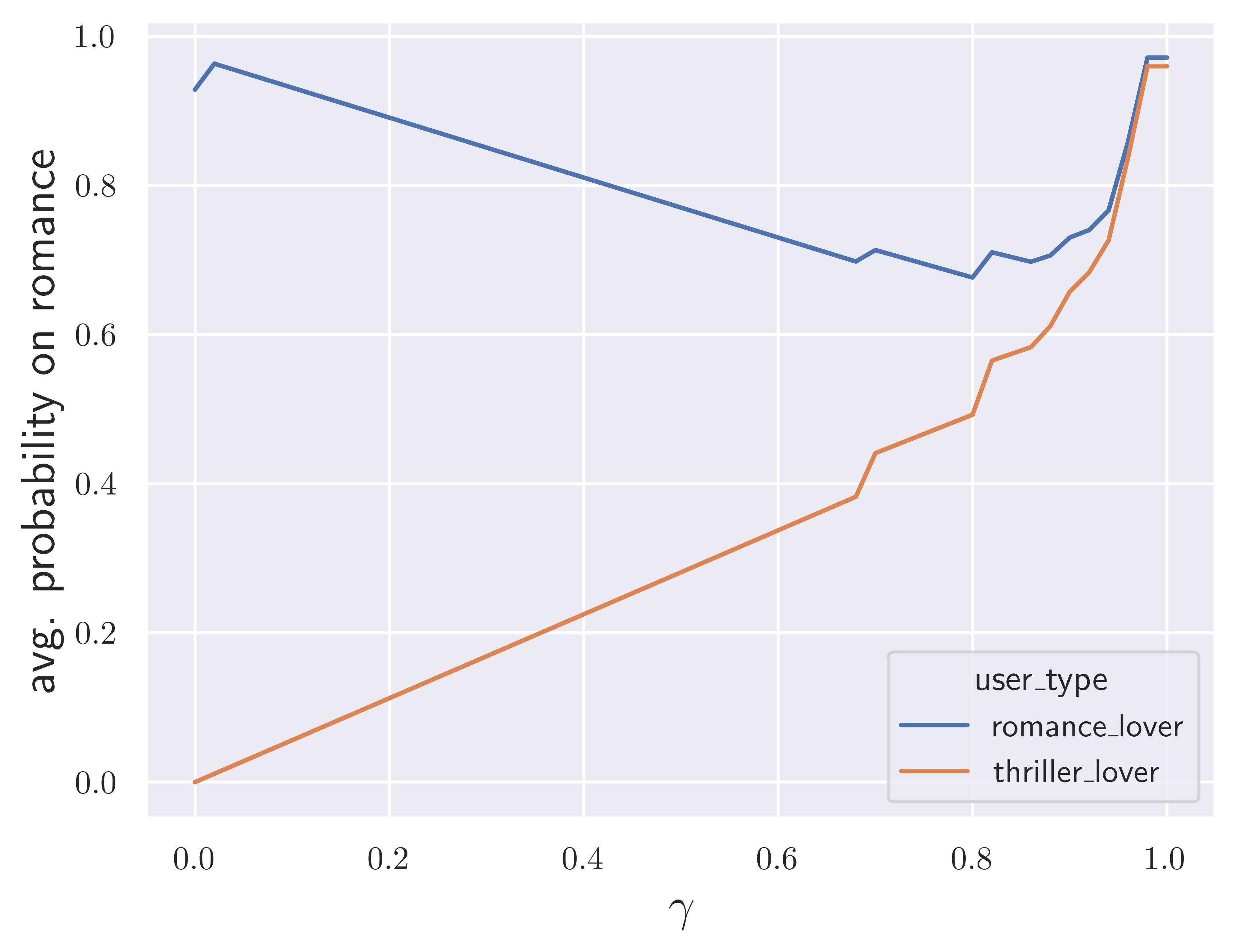}
        \caption{ Average probability placed on romance for romance- and thriller-lovers.}
        \label{fig:pcap_prob_2}
    \end{subfigure}
    \quad 
    \begin{subfigure}[b]{0.48\textwidth}
        \centering
        \includegraphics[width=\textwidth]{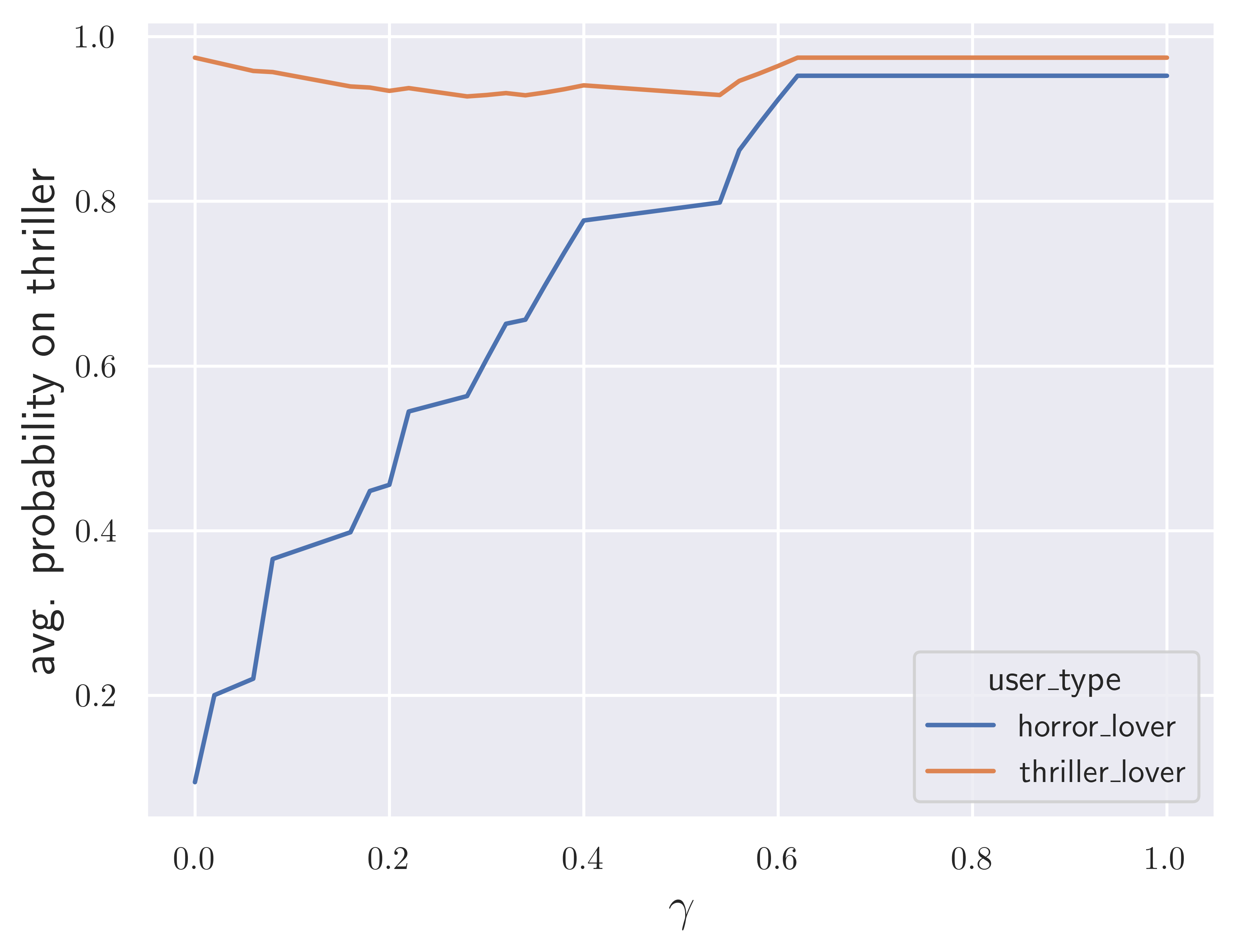}
        \caption{Average probability placed on thriller for horror- and thriller-lovers.}
        \label{fig:pcap_prob_1}
    \end{subfigure}
    \caption{Polarization cap: Content changes as a function of $\gamma$ for 2 user groups. We compute the optimal policy for 50 values of $\gamma$ equally spaced between $[0,1]$.}\label{fig:pcap_prob}
\end{figure}
In Figure~\ref{fig:pcap_prob_2}, we plot the probability placed on romance by the optimal policy (which maximizes $\sum_{i = 1}^n \vec{\mu}_i \cdot \vec{p}_i$ such that $\vec{p}_i \geq \frac{\gamma}{n} \sum_{i' = 1}^n \vec{p}_{i'}$) as a function of $\gamma$.
For comparison, we run the same experiments for two similar genres: thriller and horror. In Figure~\ref{fig:pcap_prob_1}, we plot the probability placed on thriller.

In both Figures~\ref{fig:pcap_prob_2} and ~\ref{fig:pcap_prob_1}, as $\gamma$ increases, the content recommendations become more homogeneous. However, the rates at which the recommendations 
become homogeneous are significantly different.
In Figure \ref{fig:pcap_prob_2}  where the users are polarized, the content recommendations converge slowly. It is not until $\gamma = 0.9$ that the content is completely homogeneous. 

Meanwhile, in Figure~\ref{fig:pcap_prob_1} where the groups of users are similar, the recommendations become homogeneous at a  faster rate. In this example, they converge at approximately $\gamma = 0.6$.

\begin{figure}
    \centering
    \begin{subfigure}[b]{0.48\textwidth}
        \centering
        \includegraphics[width=\textwidth]{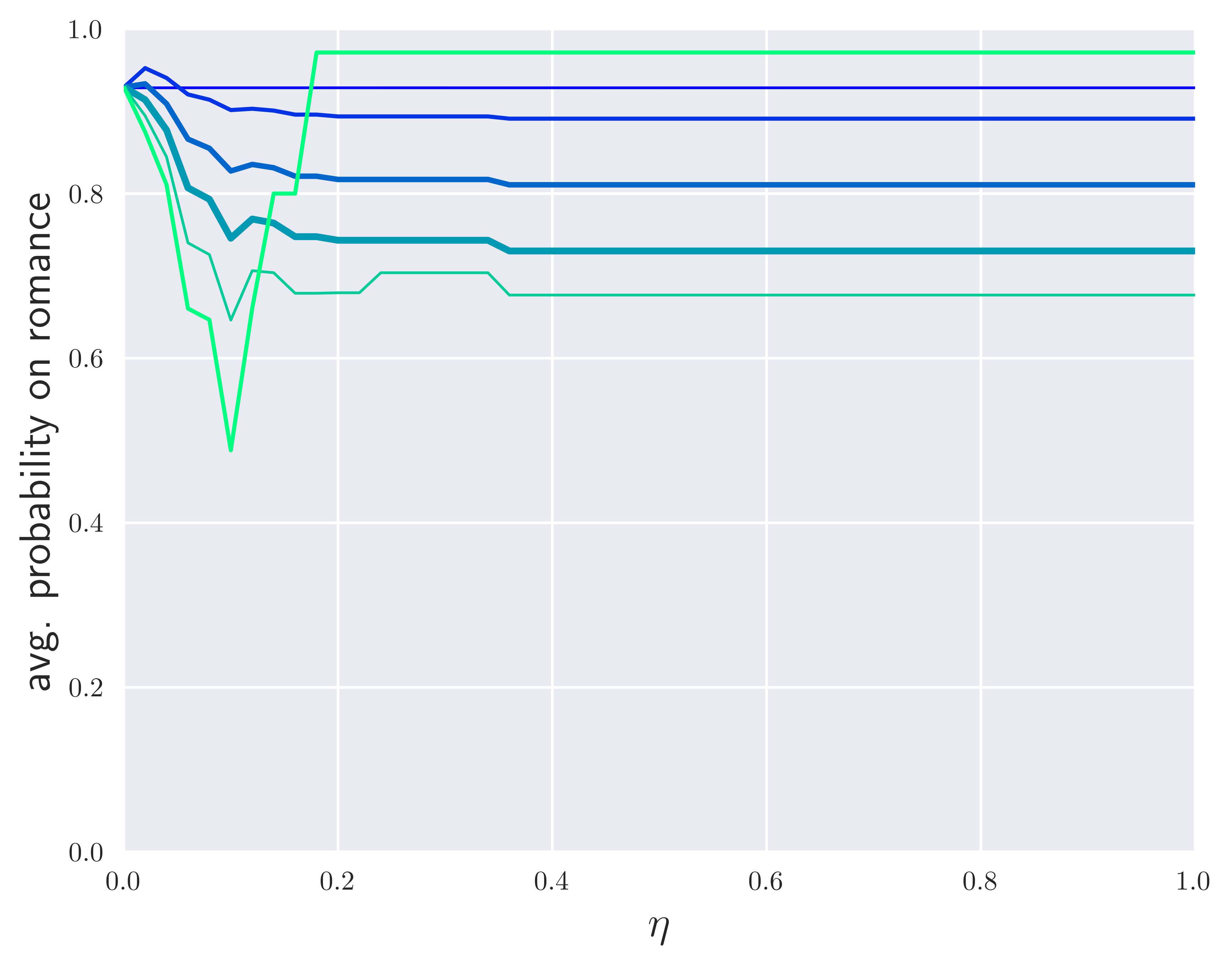}
        \caption{Probability placed on romance for romance-lovers.}
        \label{fig:ptax_prob_1}
    \end{subfigure}
    \quad 
    \begin{subfigure}[b]{0.48\textwidth}
        \centering
        \includegraphics[width=\textwidth]{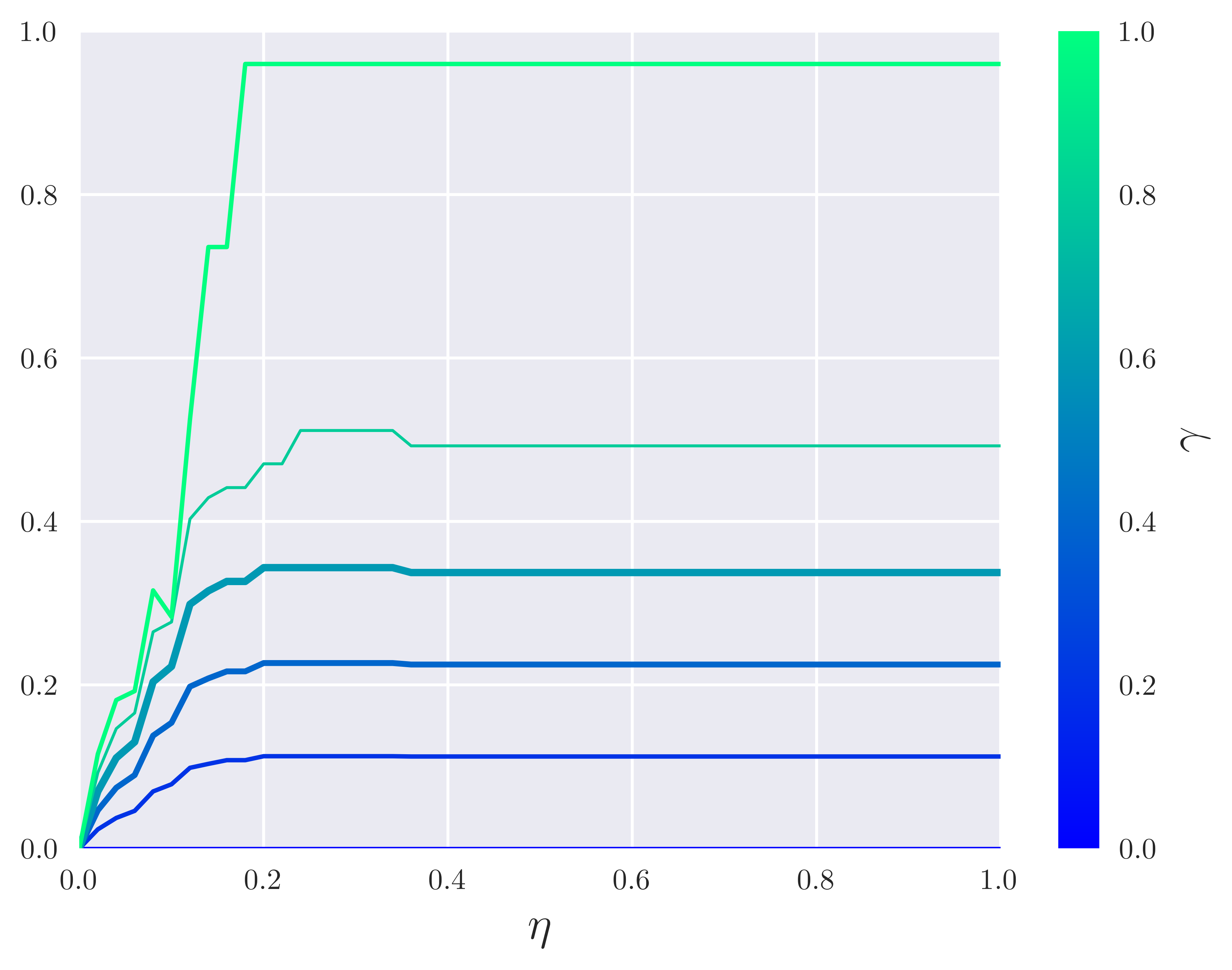}
        \caption{Probability placed on romance for thriller-lovers. \newline}
        \label{fig:ptax_prob_2}
    \end{subfigure}
    \caption{Polarization tax: Content changes as function of $\gamma$ and $\eta$ for 
        romance- and thriller-lovers. We compute the optimal policy for 6 values of $\gamma$ equally spaced between $[0,1]$ and 50 values of $\eta$ equally spaced between $[0,1]$.}
    \label{fig:ptax_prob}
\end{figure}
Under Formulation 2---where the platform is subject to a penalty (Equation~\eqref{eq:rew2})---we perform the same experiments for romance- versus thriller-lovers. These experiments are illustrated in Figure~\ref{fig:ptax_prob} where we vary both $\eta$ and $\gamma$.
As before, the content distributions converge as $\gamma$ increases. However, $\eta$ serves to modulate the impact of $\gamma$ on content recommendations.

When $\eta$ is small, the platform prefers to pay some tax to show more personalized content than they would under the hard constraint from Formulation 1. In fact, in Figure \ref{fig:ptax_prob_1}, we see that even when $\gamma =1$ (so the platform is penalized for any level of personalization), the platform prefers to pay some tax and personalize its recommendations, but for sufficiently large $\eta$ (approximately $\eta \gtrsim 0.2$), the platform switches to obeying the $\gamma$ constraint and paying no tax. For the other values of $\gamma$, the content recommendations change more gradually as $\eta$ grows. However, after a certain point ($\eta \gtrsim 0.4$), only the value of $\gamma$ leads to differences in the optimal policy.

\subsection{Effect of the polarization tax on user utility}\label{sec:user_rewards}
We next investigate the impact of the polarization penalty (Formulation 2) on the users' utility. 
We analyze the same setting from Section~\ref{sec:exp_prob} where there is a polarized society consisting of romance- and thriller-lovers.
\begin{figure}
    \centering
    \begin{subfigure}[b]{0.32\textwidth}
        \centering
        \includegraphics[width=\textwidth]{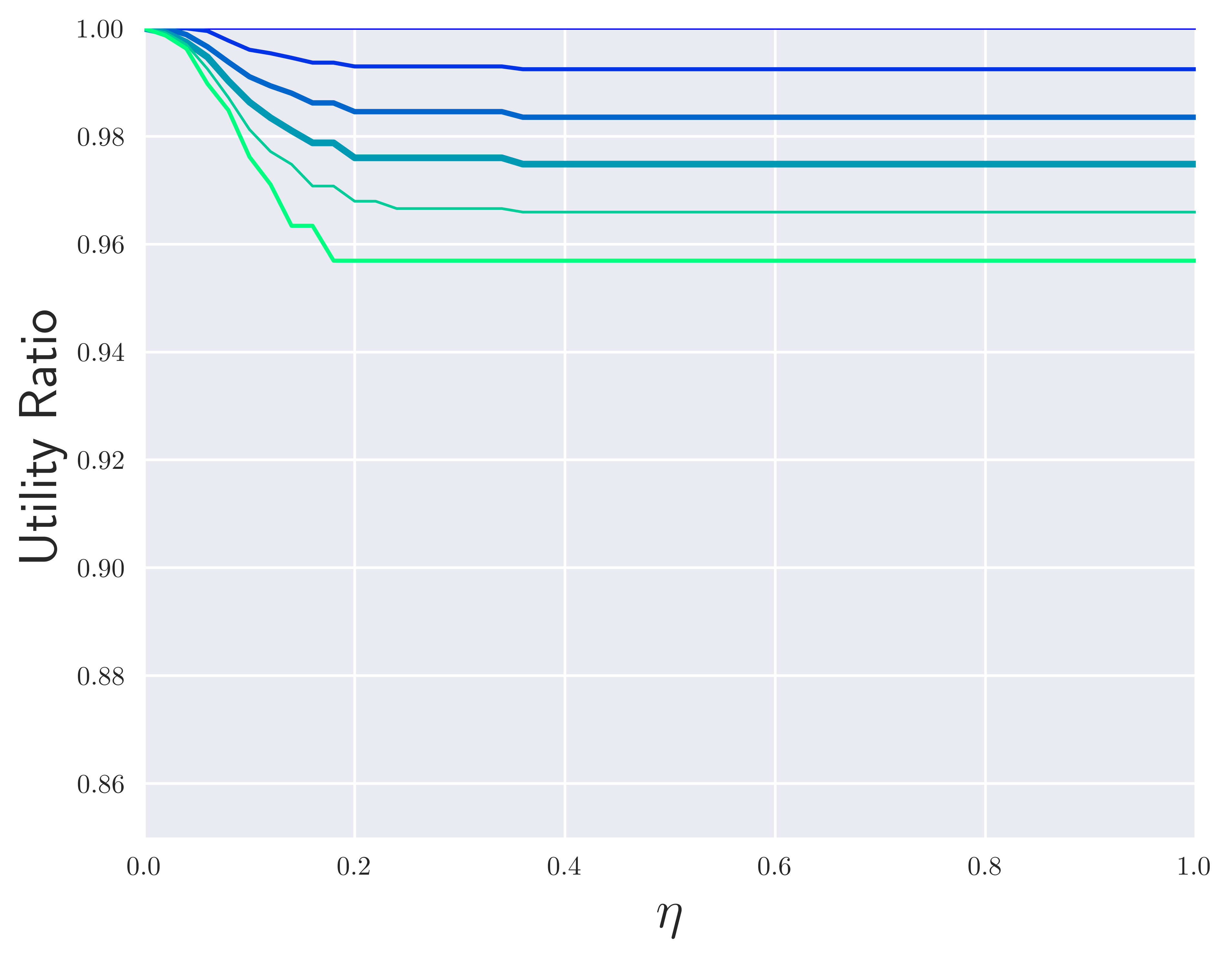}
        \caption{Romance- and thriller-lovers}
        \label{fig:ptax_util_ratio_2}
    \end{subfigure}
    \begin{subfigure}[b]{0.32\textwidth}
        \centering
        \includegraphics[width=\textwidth]{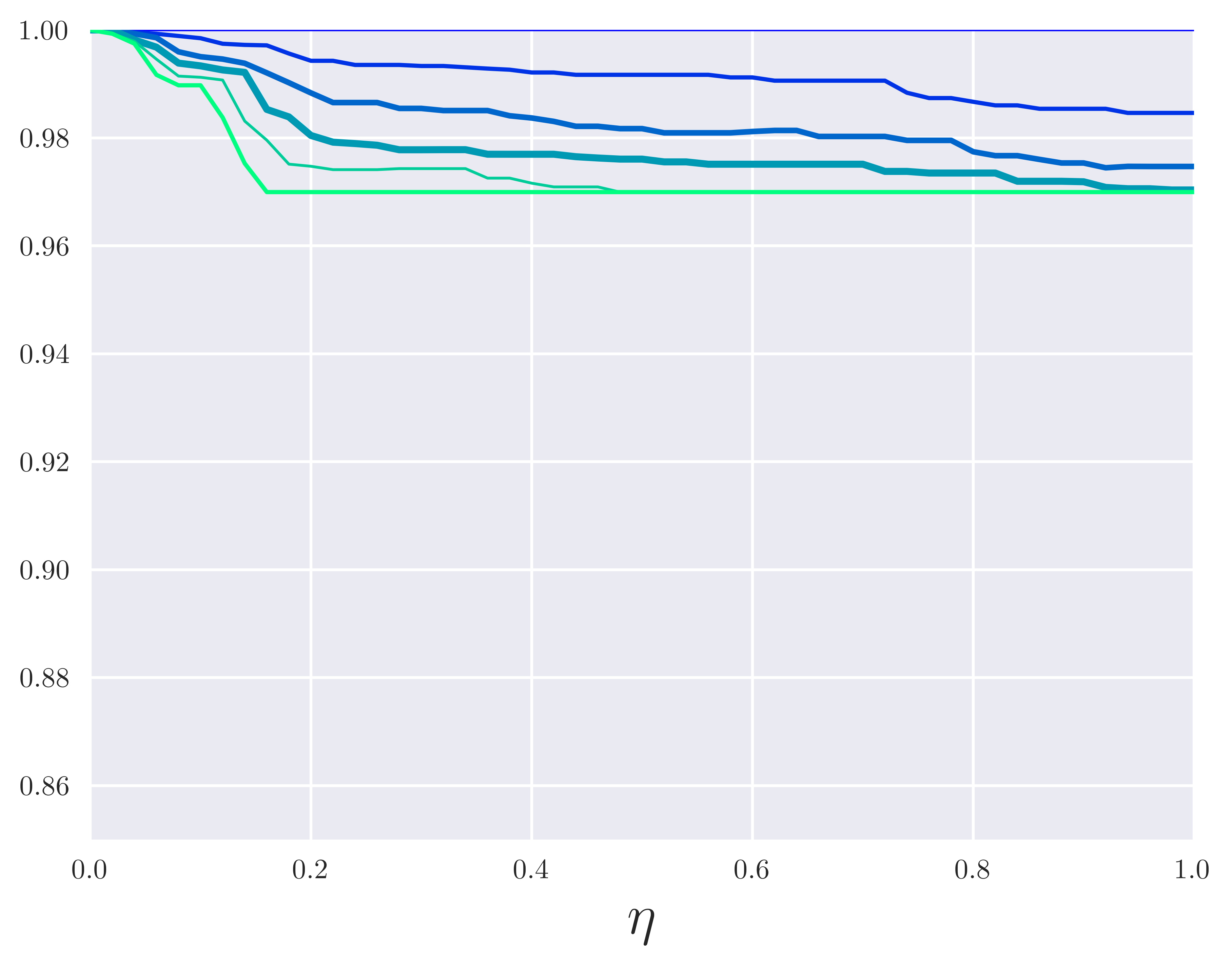}
        \caption{Horror- and thriller-lovers}
        \label{fig:ptax_util_ratio_1}
    \end{subfigure}
    \begin{subfigure}[b]{0.32\textwidth}
        \centering
        \includegraphics[width=\textwidth]{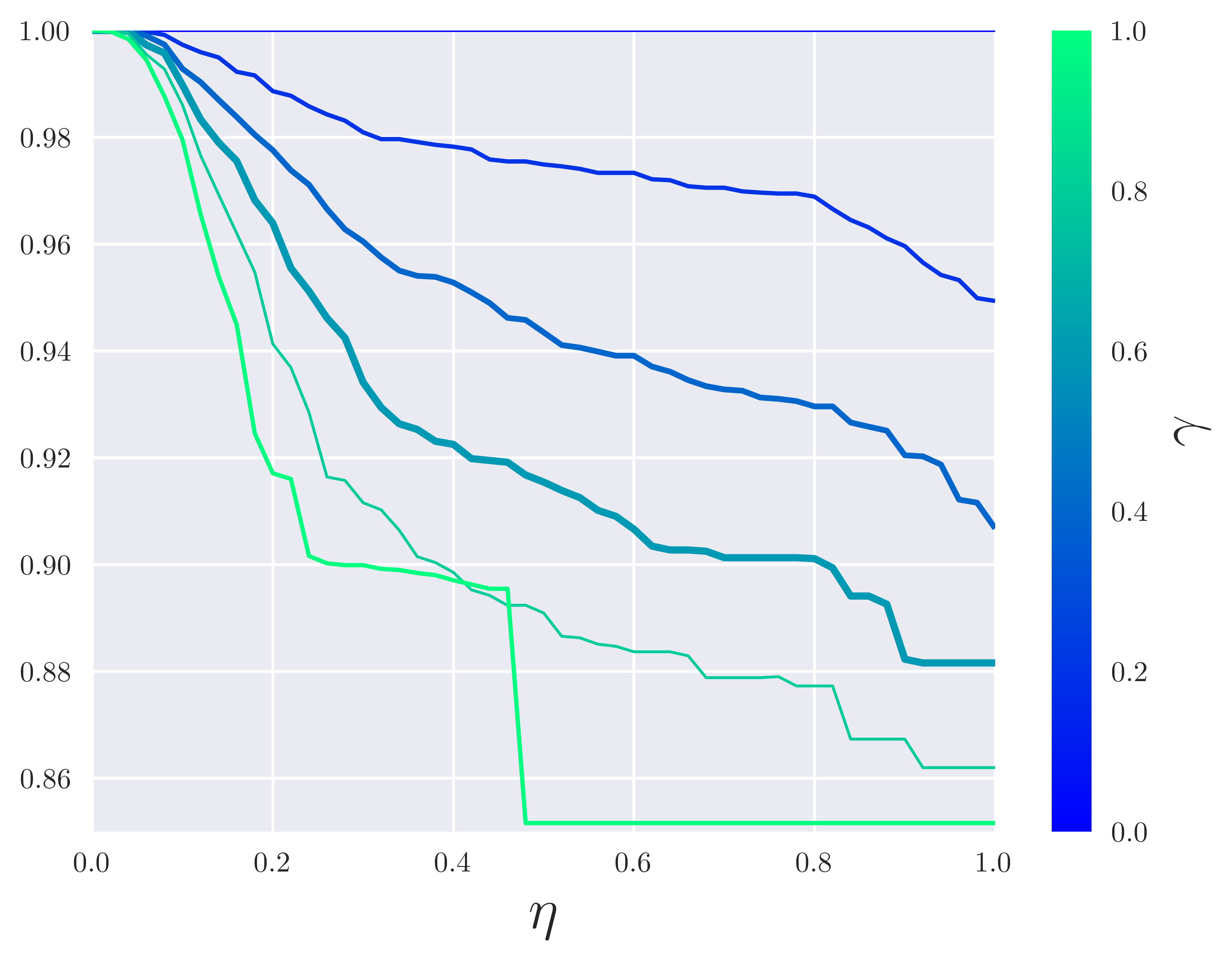}
        \caption{All user types}
        \label{fig:ptax_util_ratio_all}
    \end{subfigure}
    \caption{Multiplicative utility loss as a function of $\gamma$ and $\eta$.}
\end{figure}
Letting $(\vec{p}_i^{\gamma; \eta})_{i \in [n]}$ be the optimal policy under Equation~\eqref{eq:rew2} and $(\vec{p}_i^*)_{i \in [n]}$ policy with no penalty ($\eta = 0$), Figure~\ref{fig:ptax_util_ratio_2} plots the ratio of the users' cumulative utilities under these two policies:
$\sum \vec{\mu}_i \cdot \vec{p}_i^{\gamma; \eta}\left/\sum \vec{\mu}_i \cdot \vec{p}_i^*\right..$ Figure~\ref{fig:ptax_util_ratio_1} plots the same quantity under the homogenous society from Section~\ref{sec:exp_prob}  with only horror- and thriller-lovers.

In Figures~\ref{fig:ptax_util_ratio_2} and \ref{fig:ptax_util_ratio_1}, utility decreases as $\gamma$ and $\eta$ grow, but there is a larger utility loss for the polarized group (Figure~\ref{fig:ptax_util_ratio_2}) compared to the homogeneous group (Figure~\ref{fig:ptax_util_ratio_1}). Interestingly, in the homogenous group (Figure~\ref{fig:ptax_util_ratio_1}), utility continuously decreases as $\eta$ increases, while in the polarized group (Figure~\ref{fig:ptax_util_ratio_2}), the utility loss eventually flattens out. This is because when the population is homogeneous, as $\eta$ increases the platform only recommends one genre rather than pay the tax, even when $\gamma$ is small. However, when users are polarized (Figure~\ref{fig:ptax_prob}), the platform recommends both 
genres and pays some tax for most values of $\gamma$.  It is only when $\gamma =1$ that the platform recommends only one genre.

Figure~\ref{fig:ptax_util_ratio_all}
plots the same quantity but without restricting the genres ($\vec{\mu}_i \in [0,1]^{18}$). Since the users' preferences are more diverse, the users' cumulative utility suffers a larger but still minimal loss. This is because each user now sees a larger share of content they might not prefer since there are more groups on the platform. Finally, in Appendix~\ref{app:experiments}, we provide plots illustrating the  additive utility loss (rather than multiplicative).

\section{Conclusions and discussion}
Our work proposes a flexible approach to disincentivizing filter bubbles that adapts to the interests of the individuals on the network. Under our model, if some users are shown a particular type of content, then all users see at least a small amount of that content. We show that our model incentivizes diversity in a way that is equitable to users on the platform and %come up with
discuss algorithms for recommending content under our framework.

There remain many open questions around disincentivizing polarization in social networks. 
One might want to distinguish between the content of protected minority groups and that of hate-focused or troll groups. Our current formulation does not distinguish between these situations. One could consider a model where the penalties or cap might scale non-linearly with the size of the group, allowing for more effective moderation. In addition, there is more work to be done to understand the precise impacts of our constraints on the utility of the users and platform.  Rewards could represent the profit of the platform or the utility of its users, and our current analysis does not address this distinction. 

However, the difference could be important when there is a wealth disparity between groups and differences in utility of the users might not easily map to differences in the platform's revenue. A related direction could be to extend our model to maximize
popular and well-studied notions of fairness like \textit{Nash social welfare}.

\paragraph{Acknowledgements.} We thank Rediet Abebe for inspirational early discussions about this research direction.

\bibliographystyle{plainnat}
\bibliography{bib.bib}

\appendix
\section{Proofs about first attempt (Section~\ref{sec:first})}\label{app:first}
\naiveOpt*

\begin{proof}
First, we write the expected reward as \[\sum_{i = 1}^n \vec{p}_i \cdot \vec{\mu}_i = \sum_{i \in N}p_{i,1} + \sum_{i \not\in N}p_{i,2} = n - |N| + \sum_{i \in N}p_{i,1} - \sum_{i \not\in N}p_{i,1},\] so we can write our optimization problem as \begin{equation}\begin{array}{ll}\text{maximize} &\sum_{i \in N}p_{i,1} - \sum_{i \not\in N}p_{i,1},\\
\text{subject to} &\left|p_{i,1} - \frac{1}{n}\sum_{i' = 1}^n p_{i',0}\right| \leq \Delta.\end{array}\label{eq:z_constraint}\end{equation} 

Next, we show that there exists an optimal solution such that $p_{i,1} = q_1$ for all $i\in N$ and $p_{i,1} = q_2$ for all $i\not\in N$, for some $q_1, q_2 \in [0,1].$

\begin{claim}\label{claim:WLOG_first_attempt}
An optimal solution to Equation~\eqref{eq:z_constraint} has $p_{i,1} = q_1$ for all $i\in N$ and $p_{i,1} = q_2$ for all $i\not\in N$, for some $q_1, q_2 \in [0,1].$
\end{claim}
\begin{proof}[Proof of Claim~\ref{claim:WLOG_first_attempt}]
First, if $N = \emptyset$, then the optimal solution is to set $p_{i,1} = 0$ for all $i \in [n]$, and if $N = [n]$, then the optimal solution is to set $p_{i,1} = 1$ for all $i \in [n]$. In both of these cases, the claim holds.

Next, suppose $N \not= \emptyset$ and $N \not= [n]$.
Let $p_{1,0}, \dots, p_{n,0}$ be an optimal solution to Equation~\eqref{eq:z_constraint} and let $q_1 = \frac{1}{|N|}\sum_{i \in N}p_{i,1}$ and $q_2 = \frac{1}{n - |N|}\sum_{i \not\in N}p_{i,1}$. This is a feasible solution to Equation~\eqref{eq:z_constraint} because \[\frac{1}{n}\left(|N|q_1 + (n-|N|)q_2\right) = \frac{1}{n}\sum_{i = 1}^n p_{i,1},\] so for all $i \in N$, we have that \[-\Delta \leq \min_{j \in N} p_{j,0} - \frac{1}{n}\sum_{i' = 1}^n p_{i',0} \leq q_1 - \frac{1}{n}\left(|N|q_1 + (n-|N|)q_2\right) \leq \max_{j \in N} p_{j,0} - \frac{1}{n}\sum_{i' = 1}^n p_{i',0} \leq \Delta.\] Similarly, for all $i\not \in N$, \[-\Delta \leq \min_{j \not \in N} p_{j,0} - \frac{1}{n}\sum_{i' = 1}^n p_{i',0} \leq q_2 - \frac{1}{n}\left(|N|q_1 + (n-|N|)q_2\right) \leq \max_{j \not\in N} p_{j,0} - \frac{1}{n}\sum_{i' = 1}^n p_{i',0} \leq \Delta.\] Moreover, this solution has the same objective function value as $p_{1,0}, \dots, p_{n,0}$, so it is an optimal solution.
\end{proof}

Using this notation, we simplify the constraints by writing \[q_1 - \frac{1}{n}\sum_{i = 1}^n p_{i,1} = q_1 - \frac{1}{n}\left(|N| q_1 + (n-|N|)q_2\right) = \frac{(n-|N|)(q_1 - q_2)}{n}.\] Therefore, the constraint (Equation~\eqref{eq:z_constraint}) for any $i\in N$ becomes \begin{equation}|q_1 - q_2| \leq \frac{n\Delta}{n - |N|}.\label{eq:tighter_z}\end{equation} Similarly, we may write \[q_2 - \frac{1}{n}\sum_{i = 1}^n p_{i,1} = q_2 - \frac{1}{n}\left(|N| q_1 + (n-|N|)q_2\right) = \frac{|N|(q_2 - q_1)}{n}.\]  Therefore, the constraint (Equation~\eqref{eq:z_constraint}) for any $i\not\in N$ becomes \begin{equation}|q_1 - q_2| \leq \frac{n\Delta}{|N|}.\label{eq:tighter_o}\end{equation}

Since $|N| \geq \frac{n}{2}$, Equation~\eqref{eq:tighter_o} is tighter than Equation~\eqref{eq:tighter_z}.
Therefore, our optimization problem can be written as the LP \begin{align}\text{maximize}\qquad&g(q_1, q_2) = |N|q_1 - (n-|N|)q_2\nonumber\\
\text{subject to} \qquad &q_2 \geq q_1 - \frac{n\Delta}{|N|} \label{eq:po_big}\\
&q_2 \leq q_1 + \frac{n\Delta}{|N|} \label{eq:po_small}\\
&0 \leq q_2, q_1 \leq 1\nonumber \end{align}
The vertices $(q_1,q_2)$ of this LP polytope and their objective values $g(q_1, q_2)$ are
\begin{align*}
\text{Intersection of Equations~\eqref{eq:po_big} and \eqref{eq:po_small}: } &\text{Infeasible}\\
\text{Intersection of Equation~\eqref{eq:po_big} and } q_2 = 0\text{: } &\left(q_1^{(1)}, q_2^{(1)}\right) = \left(\frac{n\Delta}{|N|}, 0\right)\\
&g\left(q_1^{(1)}, q_2^{(1)}\right) = n\Delta\\
\text{Intersection of Equation~\eqref{eq:po_big} and } q_2 = 1\text{: } &\text{Infeasible}\\
\text{Intersection of Equation~\eqref{eq:po_big} and } q_1 = 0\text{: } &\text{Infeasible}\\
\text{Intersection of Equation~\eqref{eq:po_big} and } q_1 = 1\text{: } &\left(q_1^{(2)}, q_2^{(2)}\right) = \left(1, 1 - \frac{n\Delta}{|N|}\right)\\
&g\left(q_1^{(2)}, q_2^{(2)}\right) = 2|N| - n + n\Delta\left(\frac{n}{|N|} -1\right)\\
\text{Intersection of Equation~\eqref{eq:po_small} and } q_2 = 0\text{: } &\text{Infeasible}\\
\text{Intersection of Equation~\eqref{eq:po_small} and } q_2 = 1\text{: } &\left(q_1^{(3)}, q_2^{(3)}\right) = \left(1 - \frac{n\Delta}{|N|}, 1\right)\\
&g\left(q_1^{(3)}, q_2^{(3)}\right) = 2|N| - n + n\Delta\\
\text{Intersection of Equation~\eqref{eq:po_small} and } q_1 = 0\text{: } &\left(q_1^{(4)}, q_2^{(4)}\right) = \left(0, \frac{n\Delta}{|N|}\right)\\
&g\left(q_1^{(4)}, q_2^{(4)}\right) = -\frac{n\Delta(n-|N|)}{|N|}\\
\text{Intersection of Equation~\eqref{eq:po_small} and } q_1 = 1\text{: } &\text{Infeasible.}
\end{align*}
Finally, we have vertices $\left(q_1^{(5)}, q_2^{(5)}\right) = (0,0)$ with $g\left(q_1^{(5)}, q_2^{(5)}\right) = 0$ and $\left(q_1^{(6)}, q_2^{(6)}\right) = (1,1)$ with $g\left(q_1^{(6)}, q_2^{(6)}\right) = 2|N|-n$. Since $\Delta < \frac{|N|}{n}$, $(q_1,q_2) = (0,1)$ and $(q_1,q_2) = (1,0)$ are infeasible. Since $|N| \geq \frac{n}{2}$, $\left(q_1^{(3)}, q_2^{(3)}\right) = \left(1 - \frac{n\Delta}{|N|}, 1\right)$ maximizes the objective value, which implies the lemma statement.
\end{proof}

\section{Proofs about our more equitable approach (Section~\ref{sec:equitable})}\label{app:equitable}
\interpretation*

\begin{proof}
    Our goal is to find distributions $\vec{p}_1, \dots, \vec{p}_n \in \cP^1$ to the optimization problem \begin{equation}\begin{array}{ll}\text{maximize} &\sum_{i = 1}^n \vec{p}_i \cdot \vec{\mu}_i = \sum_{i \in N} p_{i,1} + \sum_{i \not\in N} p_{i,2}\\
    \text{such that} &
        p_{i,1} \geq \frac{\gamma}{n} \sum_{i' = 1}^n p_{i',1}, \forall i \in [n]\\
        & p_{i,2} \geq \frac{\gamma}{n} \sum_{i' = 1}^n p_{i',2}, \forall i \in [n]
    \end{array}\label{eq:interpretation_LP}\end{equation} We claim that without loss of generality, we may set $p_{i,1} = q_1$ for all $i \in N$ and $p_{i,2} = q_2$ for all $i \not\in N$, for some $q_1, q_2 \in [0,1].$

    \begin{claim}\label{claim:WLOG}
An optimal solution to Equation~\eqref{eq:interpretation_LP} has $p_{i,1} = q_1$ for all $i \in N$ and $p_{i,2} = q_2$ for all $i \not\in N$, for some $q_1, q_2 \in [0,1].$
\end{claim}
\begin{proof}[Proof of Claim~\ref{claim:WLOG}]
First, if $N = \emptyset$, then the optimal solution is to set $p_{i,2} = 1$ for all $i \in [n]$, and if $N = [n]$, then the optimal solution is to set $p_{i,1} = 1$ for all $i \in [n]$. In both of these cases, the claim holds.

Next, suppose $N \subset [n]$ and $N \not=\emptyset$.
Let $\vec{p}_1, \dots, \vec{p}_n \in \cP^1$ be an optimal solution to Equation~\eqref{eq:interpretation_LP} and let $q_1 = \frac{1}{|N|}\sum_{i \in N}p_{i,1}$ and $q_2 = \frac{1}{n - |N|}\sum_{i \not\in N}p_{i,2}$. This is a feasible solution to the constraints in Equation~\eqref{eq:interpretation_LP} because \begin{align*}q_1 &= \frac{1}{|N|} \sum_{i \in N} p_{i,1} \geq \frac{1}{|N|} \sum_{i \in N} \frac{\gamma}{n} \sum_{i' = 1}^n p_{i',1} = \frac{\gamma}{n} \sum_{i = 1}^n p_{i,1} = \frac{\gamma}{n} \left(\sum_{i \in N} p_{i,1} + \sum_{i \not\in N} \left(1 - p_{i,1}\right)\right)\\&= \frac{\gamma}{n}\left(|N|q_1 + (n-|N|)\left(1 - q_2\right)\right).\end{align*} Similarly, \begin{align*}q_2 &= \frac{1}{|N|} \sum_{i \in N} p_{i,2} \geq \frac{1}{|N|} \sum_{i \in N} \frac{\gamma}{n} \sum_{i' = 1}^n p_{i',2} = \frac{\gamma}{n} \sum_{i = 1}^n p_{i,2} = \frac{\gamma}{n} \left(\sum_{i \in N} \left(1 - p_{i,1}\right) + \sum_{i \not\in N} p_{i,1}\right)\\&= \frac{\gamma}{n}\left(|N|\left(1 - q_1\right) + (n-|N|)q_2\right).\end{align*}
Moreover, the objective functions are the same because \[|N|q_1 + (n-|N|)q_2 = \sum_{i \in N} p_{i,1} + \sum_{i \not\in N} p_{i,2}.\] Therefore, the claim holds.
\end{proof}
Based on Claim~\ref{claim:WLOG}, we may write our optimization problem as 
\begin{align*}\text{maximize} \quad &|N|q_1 + (n-|N|)q_2 \nonumber\\
    \text{such that} \quad &
        q_1 \geq \frac{\gamma}{n} \left(|N|q_1 + (n-|N|)\left(1 - q_2\right)\right)\\
        &1 - q_1 \geq \frac{\gamma}{n} \left(|N|\left(1 - q_1\right) + (n-|N|)q_2\right)\\
        &q_2 \geq \frac{\gamma}{n} \left(|N|\left(1 - q_1\right) + (n-|N|)q_2\right)\\
        &1 - q_2 \geq \frac{\gamma}{n} \left(|N|q_1 + (n-|N|)\left(1 - q_2\right)\right)\\
        &q_1, q_2 \in [0,1]. \end{align*}

Rearranging terms, our optimization problem is
\begin{align}\text{maximize} \quad &g(q_1, q_2) = |N|q_1 + (n-|N|)q_2 \nonumber\\
    \text{such that} \quad &
        q_1 \geq \frac{\gamma(n-|N|)}{n - \gamma|N|}(1 - q_2)\label{eq:h1}\\
        &q_1 \leq 1 - \frac{\gamma(n-|N|)}{n - \gamma|N|} \cdot q_2\label{eq:h2}\\
        &q_1 \geq 1 - \frac{n - \gamma(n - |N|)}{n - \gamma|N|} \cdot q_2\label{eq:h3}\\
        &q_1 \leq \frac{n - \gamma(n - |N|)}{\gamma|N|}(1 - q_2)\label{eq:h4}\\
        &q_1, q_2 \geq 0 \text{ and } q_1, q_2 \leq 1.\label{eq:h5}
   \end{align}

To analyze the corners of this LP polytope, we identify where the eight hyperplanes in Equations~\eqref{eq:h1}-\eqref{eq:h5} intersect. Note that Equations~\eqref{eq:h1} and \eqref{eq:h2} are parallel, so they don't intersect. The same is true of Equations~\eqref{eq:h3} and \eqref{eq:h4}. Moreover, $q_1 = 0$ if and only if $q_2 = 1$ by Equations~\eqref{eq:h1} and \eqref{eq:h4}, so we ignore the intersections with $q_1 = 0$ and $q_2 = 1$.
This leads to the following corners $(q_1, q_2)$ with objective values $g(q_1, q_2)$:
\begin{align}
\text{Intersection of \eqref{eq:h1} and \eqref{eq:h3}: } &\left(q_1^{(1)}, q_2^{(1)}\right) = \left(\frac{\gamma(n-|N|)}{n}, \frac{\gamma|N|}{n}\right)\nonumber\\
& g\left(q_1^{(1)}, q_2^{(1)}\right) = \frac{2\gamma|N|(n-|N|)}{n}\nonumber\\
\text{Intersection of \eqref{eq:h1} and \eqref{eq:h4}: } &\left(q_1^{(2)}, q_2^{(2)}\right) = \left(0,1\right)\nonumber\\
& g\left(q_1^{(2)}, q_2^{(2)}\right) = n - |N|\nonumber\\
\text{Intersection of \eqref{eq:h1} and }q_2 = 0\text{: } &\left(q_1^{(3)}, q_2^{(3)}\right) = \left(\frac{\gamma(n-|N|)}{n - \gamma|N|}, 0\right)\nonumber\\
& g\left(q_1^{(3)}, q_2^{(3)}\right) = \frac{\gamma(n-|N|)|N|}{n - \gamma|N|}\nonumber\\
\text{Intersection of \eqref{eq:h1} and }q_1 = 1\text{: } &\text{Not feasible}\nonumber\\
\text{Intersection of \eqref{eq:h2} and \eqref{eq:h3}: } &\left(q_1^{(4)}, q_2^{(4)}\right) = (1, 0)\label{eq:2and3}\\
& g\left(q_1^{(4)}, q_2^{(4)}\right) = |N|\nonumber\\
\text{Intersection of \eqref{eq:h2} and \eqref{eq:h4}: } &\left(q_1^{(5)}, q_2^{(5)}\right) = \left(1 - \gamma\left(1 - \frac{|N|}{n}\right), 1 - \frac{\gamma|N|}{n}\right)\nonumber\\
& g\left(q_1^{(5)}, q_2^{(5)}\right) = n - \frac{2\gamma|N|(n-|N|)}{n}\nonumber\\
\text{Intersection of \eqref{eq:h2} and } q_2=0 \text{: }&(q_1, q_2) = (1,0) \text{ as in Equation~\eqref{eq:2and3}}\nonumber\\
\text{Intersection of \eqref{eq:h2} and } q_1=1 \text{: }&(q_1, q_2) = (1,0) \text{ as in Equation~\eqref{eq:2and3}}\nonumber\\
\text{Intersection of \eqref{eq:h3} and } q_2=0 \text{: }&(q_1, q_2) = (1,0) \text{ as in Equation~\eqref{eq:2and3}}\nonumber\\
\text{Intersection of \eqref{eq:h3} and } q_1=1 \text{: }&(q_1, q_2) = (1,0) \text{ as in Equation~\eqref{eq:2and3}}\nonumber\\
\text{Intersection of \eqref{eq:h4} and } q_2=0 \text{: }&\text{Not feasible.}\nonumber\\
\text{Intersection of \eqref{eq:h4} and } q_1=1 \text{: }&\text{Not feasible.}\nonumber
\end{align}
For $\gamma \leq\frac{1}{2}$, the optimum is achieved at $\left(q_1^{(5)}, q_2^{(5)}\right) = \left(1 - \gamma\left(1 - \frac{|N|}{n}\right), 1 - \frac{\gamma|N|}{n}\right)$, so the lemma holds.
\end{proof}

\section{Proofs about the Formulation 1 regret upper bound when $\gamma < 1$ (Section~\ref{sec:regret_small_gamma})}\label{app:regret_small_gamma}
\begin{algorithm}[t]
     \caption{Multi-agent UCB (defined by parameter $\delta$)}\label{alg:ucb.f2}
     \begin{algorithmic}[1]
             \Require Failure probability $\delta \in (0,1)$
             \State Set  $N_{i, j}(0)= 0,~ \forall i \in [n], j\in [k]$;    $\vec{\hat{\mu}}_i^{(0)} = \vec{0}, ~~ \forall i \in [n]$
             \For{$t \in \{1, \dots, T\}$}
                 \If {$t \in \{ 1, \ldots, k\}$}
                    \State Set $\vecp_i^{(t)} = \vec{e}_t$
                 \Else
                    \State Set $\left(\vecp_i^{(t)}\right)_{i \in [n]} = \amax\limits_{(\vecp_i)_{i \in [n]} \in S} \sum\limits_{i = 1}^n \vecp_i \cdot \vec{\hat{\mu}}_i^{(t-1)}$
                 \EndIf
                 \State Draw an arm $j_i^{(t)} \sim \vecp_i^{(t)} ~~ \forall i \in [n]$
                 \State Receive reward $r_i^{(t)} \sim \cD_{i, j_i^{(t)}}$
                 \State For all $i \in [n]$, set $N_{i,  j_{i}^{(t)}}(t) = N_{i,  j_{i}^{(t)}}(t-1) + 1$%, ~~ \forall i \in [n]$
                 \Comment{Increment the counter for arm $j_i^{(t)}$}
                 \State Set $N_{i,  j}(t) = N_{i,  j}(t-1) , ~~ \forall i \in [n]$ and  $j \neq j_{i}^{(t)}$ \Comment{Do not increment the other counters}
                 \State Set $\beta _{i, j}^{(t)} = \sqrt{\frac{1}{N_{i,j}(t)}\log\frac{2Tnk}{\delta}}, ~~ \forall i \in [n], j \in [k]$ \Comment{Define confidence intervals}
                 \State $\hat{\mu}_{i, j}^{(t)} = \frac{1}{N_{i, j}(t)}\sum\limits_{\tau = 1}^t r_i^{(\tau)}\ind{j_{i}^{(\tau)} = j} + \beta_{i, j}^{(t)}$, $~~ \forall i \in [n], j \in [k]$ \Comment{Estimate mean rewards}
             \EndFor
         \end{algorithmic}
 \end{algorithm}

\begin{claim}\label{clm:ub.f2.1}
    With probability $1 - \delta$, for all $ i \in [n], t \in [T]$ and $j \in [k]$, \[\vec{\hat{\mu}}_{i,j}^{(t)} \geq \vec{\mu}_{i, j} \geq \vec{\hat{\mu}}_{i,j}^{(t)} - 2\vec{\beta}_{u, i}^{(t)}\]
\end{claim}

\begin{proof}
    Consider a fixed iteration $t$. For $i \in [n]$ and $j \in [k]$  and $\ell \in [t]$, let $\tau  = \inf_{s}\{N_{i,j}(s) \geq \ell\}$ and $ \hat{v}^{\ell}_{i,j} = \frac{1}{\ell}\sum\limits_{s = 1}^{\tau } r_{i,j}^{(s)}\ind{j_{i}^{(s)} = j}$. Since  the  $\hat{v}^{\ell}_{i,j}$ are independent in $\ell, i,$ and $j$ applying Hoeffding's inequality with parameter $\delta' = \delta/(tnk)$, with probability greater than $1-\delta'$, 
    \[ 
            |\hat{v}^{\ell}_{i,j}- \mu_{i,j}| \leq \sqrt{\frac{\log(2/\delta')}{\ell}}
    \]
 Applying the union bound for all $\forall \ell \in [t], \forall i \in [n], \forall  j \in [k]$ we have that:
 
 \begin{align*}
     &\P\left(\exists \ell \in [t], \exists i \in [n], \exists  j \in [k]:   |\hat{v}^{\ell}_{i,j} - \mu_{i,j}| \geq \sqrt{\frac{\log(2/\delta')}{\ell}}\right) \\
      \leq \, &\sum_{\ell =1}^{t}\sum_{i=1}^n \sum_{j=1}^k \P\left(|\hat{v}^{\ell}_{i,j} - \mu_{i,j}| \geq \sqrt{\frac{\log(2/\delta')}{N_{i,j}(t)}}\right) \\
      \leq \, &tnk\delta' = \delta
 \end{align*}
Thus with probability at least $1-\delta$, $\forall \ell \in [t], \forall i \in [n], \forall  j \in [k]$,
  \[|\hat{v}_{i,j}^{\ell} - \mu_{i,j}| \leq \sqrt{\frac{\log(2tnk/\delta)}{\ell}}.\]

 Since $\hat{\mu}_{i,j}^t = \hat{v}_{i,j}^{N_{i,j}(t)}  + \beta_{i,j}^{(t)}$ and  $N_{i,j}(t) \in [t]  ~\forall i,j $ with probability at least $1-\delta$ 
\[\hat{\mu}_{i,j}^t  -2\beta_{i,j}^{(t)} \leq \mu_{i,j} \leq \hat{\mu}_{i,j}^t.\]
\end{proof}

\nUCB*

\begin{proof}
    Fix a timestep $t \in [T]$
    \begin{align*}
            \sum_{i \in [n]}(\vecp_i^* \cdot \vec{\mu}_i - \vecp_i^{(t)}\cdot \vec{\mu}_i )  &= \sum_{i \in [n]}(\vecp_i^*\cdot \vec{\mu}_i - \vecp_i^{(t)}\cdot \vec{\hat{\mu}}_i^{(t)} +  \vecp_i^{(t)}\cdot \vec{\hat{\mu}}_i^{(t)} -\vecp_i^{(t)}\cdot \vec{\mu}_i) \\
                    &\leq \sum_{i \in [n]}(\vecp_i^*\cdot \vec{\mu}_i - \vecp_i^{*}\cdot \vec{\hat{\mu}}_i^{(t)} +  \vecp_i^{(t)}\cdot \vec{\hat{\mu}}_i^{(t)} -\vecp_i^{(t)}\cdot \vec{\mu}_i)
        \end{align*}
    By Claim \ref{clm:ub.f2.1}, $\vecp \cdot \vec{\mu}_i \leq \vecp \cdot \vec{\hat{\mu}}_i^{(t)}$ $\forall i \in [n]$ and all $\vecp \in \R_{\geq 0}^k$
    \begin{align*}
          \sum_{i \in [n]}(\vecp_i^* \cdot \vec{\mu}_i- \vecp_i^{(t)}\cdot \vec{\mu}_i )  &\leq \sum_{i\in [n]}(\vecp_i^* \cdot \vec{\hat{\mu}}_i^{(t)} - \vecp_i^{(t)}\cdot \vec{\mu}_i )  \\
                &\leq \sum_{i \in [n]} \vecp_i^{(t)} \cdot \vec{\beta}_i^{(t)}
        \end{align*}

 Thus \[  R_T \leq \sum_{t =1}^T\sum_{i \in [n]}\vecp_i^{(t)} \cdot \vec{\beta}_i^{(t)}\]
            
  Let $\mathcal{F}_{i, t-1}$ denote the canonical filtration $\sigma((X_{i,s},\vecp_i^{(s)}): 0\leq s < t)$ on the choice of $\vecp_i^{(t)}$ and let $j_{i}^{(t)'}$ be a random variable distributed as $\vecp_i^{(t)} \mid \mathcal{F}_{i,t-1}$ and conditionally independent from $j_{i,t}^{(t)}$, i.e.$~ j_{i}^{(t)'} \perp j_{i}^{(t)}\mid \mathcal{F}_{i,t-1}$. Note that by definition the following equality holds: 

  \begin{equation*}
      \E_{j_{i}^{(t)}\sim \vecp_{i}^{(t)}}[\beta_{i, j_{i}^{(t)}}^{(t)}] = \E_{j_{i}^{(t)'} \sim  \vecp_{i}^{(t)}}[\beta_{i, j_{i}^{(t)'}}^{(t)'} \mid \mathcal{F}_{i, t-1}].
  \end{equation*}

  Consider the following random variables $A_{i,t} =  \E_{j_{i}^{(t)'} \sim  \vecp_{i}^{(t)}}[\beta_{i, j_{i}^{(t)'}} \mid \mathcal{F}_{i, t-1} ]- \beta_{i,j_i^{(t)}}(t)$. Note that $M_{i,t} = \sum_{s=1}^t A_{i,s}$ is a martingale. Since $|A_t| \leq 2\sqrt{2 \log(Tnk/\delta)}$,  using this as the bound in Azuma-Hoeffding and taking a union bound over $i \in [n]$ and $t \in T$ implies that with probability at least $1- \delta$,
  
  \[R_T = \sum_{t =1}^T\sum_{i =1}^n \vecp_i^{(t)} \cdot \vec{\beta}_i^{(t)} \leq \sum_{t =1}^T\sum_{i =1}^n  \vec{\beta}_{i, j_{i}^{(t)}}^{(t)} + nT\sqrt{\frac{1}{T}\log\left(\frac{Tnk}{\delta}\right)\log\left(\frac{1}{\delta}\right)}\]  
  
   \[= \sum_{t =1}^T\sum_{i =1}^n  \vec{\beta}_{i, j_{i}^{(t)}}^{(t)} + n\sqrt{T\log\left(\frac{Tnk}{\delta}\right)\log\left(\frac{1}{\delta}\right)}\] 
   
  Now we bound  $\sum\limits_{i = 1}^n\sum\limits_{t=1}^T \vec{\beta}_{i, j_{i}^{(t)}}^{(t)}$, 
  
\[\sum_{t=1}^T \sum_{i =1}^n \vec{\beta}_{i, j_{i}^{(t)}}^{(t)} = \sum_{t=1}^T \sum_{i =1}^n \sum_{j =1}^k  \beta_{i,j}^{(t)}\ind{j_i^{(t)} = j}  \]

For fixed $i, j$ 
\[ \sum_{t=1}^T \beta_{i,j}^{(t)}\ind{j_i^{(t)} = j} = \sqrt{\log(Tnk/\delta)}\sum_{t=1}^{N_{i,j}(T)}1/\sqrt{t} \leq 2\sqrt{N_{i,j}(T)\log(Tnk/\delta)}\] 
Therefore 
\begin{align*}
    \sum_{i \in [n]}\sum_{t=1}^T \vec{\beta}_i^{(t)} &\leq 2\sum_{i =1}^n \sum_{j =1}^k \sqrt{N_{i,j}(T)\log(Tnk/\delta)} \\
    &\leq  2\sum_{i =1}^n \sqrt{k \sum_{j=1}^k N_{i,j}(T)\log(Tnk/\delta)} \\
    &=  2\sum_{i =1}^n \sqrt{k T \log(Tnk/\delta)}  \\
    &= 2n \sqrt{k T \log(Tnk/\delta)} 
\end{align*}
Where the second line follows from the concavity of $\sqrt{\cdot}$ and the penultimate line follows from the fact that $\sum\limits_{j=1}^k N_{i,j}(T) = T$.

The result then follows by setting $\delta = \frac{1}{nT}.$
\end{proof}

\section{Proofs about the Formulation 1 regret upper bound when $\gamma = 1$ (Section~\ref{sec:regret_gamma_one})}\label{app:regret_gamma_one}
In this section, the distribution $\cD_j = \sum_{i=1}^n \cD_{i,j}$ is supported on $[0,n]$ instead of $[0,1]$ and we denote the average reward for arm $j \in [k]$ as $\mu_j = \sum_{i=1}^n \mu_{i,j}$.

\begin{definition}[median-of-means estimator~\citep{Alon1996:moments}]\label{def:median}
Let $\delta \in (0,1)$ and $X_{1}, \ldots, X_{T}$ be i.i.d random variables with mean $\E [X] = \mu$ and variance $\E|X - \mu|^2 = \sigma^2$. Let $m = \lfloor{8\log(1/\delta)  \wedge T/2} \rfloor$ and $t = \lfloor T/m \rfloor$. Let $\bar{\mu}^1, \dots, \bar{\mu}^m$ be $m$ empirical mean estimates, each one calculated on $t$ data points as follows:
        \[
        \bar{\mu}^1 = \frac{1}{t}\sum_{s = 1}^{t} X_{s} ~, ~ \bar{\mu}^2 = \frac{1}{t}\sum_{s = t+1}^{2t}X_{s} ~,  \ldots , ~ \bar{\mu}^m= \frac{1}{t}\sum_{s = (m-1)t+1}^{mt}X_{s}.
        \]
        The \emph{median-of-means estimator} $\bar{\vec{\mu}}(T, \delta)$  is the median of these $m$ empirical means.
\end{definition}

\begin{algorithm}[t]
     \caption{Robust-UCB (defined by parameter $\delta$)}\label{alg:robustucb.f2}
     \begin{algorithmic}[1]
             \Require Failure probability $\delta \in (0,1)$, median-of-means estimator $\vec{\bar{\mu}}(t, \delta)$
             \State Set  $N_{j}(0)= 0,$ $\hat{\mu}_j^{(0)}= 0$ $~ \forall j\in [k]$  
             \For{$t \in \{1, \dots, T\}$}
                 \If {$t \in \{ 1, \ldots, k\}$}
                    \State Set $\vecp^{(t)} = \vec{e}_t$
                 \Else
                    \State Set $\vecp^{(t)}= \amax\limits_{\vecp \in \cP^{k-1}} \vecp \cdot \vec{\hat{\mu}}^{(t-1)}$
                 \EndIf
                 \State Draw an arm $j^{(t)} \sim \vecp^{(t)}$
                 \State Receive reward $r^{(t)} \sim \cD_{j^{(t)}}$
                 \State Set $N_{j^{(t)}}(t) = N_{j^{(t)}}(t-1) + 1$%, ~~ \forall i \in [n]$
                 \Comment{Increment the counter for arm $j^{(t)}$}
                 \State Set $N_{j}(t) = N_{j}(t-1) , ~~ \forall j \neq j^{(t)}$ \Comment{Do not increment the other counters}
                 \State Set $\beta _{j}^{(t)} = \sqrt{\frac{24n}{N_{j}(t)}\log\frac{Tk}{\delta}}, ~~ \forall j \in [k]$ \Comment{Define confidence intervals}
                 \State $\hat{\mu}_{j}^{(t)} = \vec{\bar{\mu}}_j(N_j(t), \delta) + \beta_{j}^{(t)}$, $~~ \forall  j \in [k]$  \Comment{Get mean rewards estimates}
             \EndFor
         \end{algorithmic}
 \end{algorithm}

\regretGammaOne*

\begin{proof}
    For all $t \in [T] $, $j \in [k]$ the median-of-mean estimate at time $t$,  $\hat{\mu}_j^{(t)}$ with probability at least $1-\delta$ we have
    \[
    |\hat{\mu}_j^{(t)} - \mu_j | \leq \sqrt{\frac{24n\log(kT/\delta)}{T}}
    \]
    
    By applying Lemma \ref{lem:median-robust} with $\sigma^2_j = n/4$ which is justified by Claim\ref{clm:var_low} and then taking  a union bound over all arms $j \in [k]$.
    
     This event and Proposition 1 in \citet{Bubeck13:Bandits} imply the desired regret upper bounds for the Robust-UCB with median-of-means estimator. The result then follows by setting $\delta = \frac{1}{nT}.$
\end{proof}
    
        \begin{lemma}\label{lem:median-robust}
        Let $\delta \in (0,1)$. Let $X_{j,1}, \ldots, X_{j,T}$ be i.i.d random variables with mean $\E [X_j] = \mu_j$ and $\E|X_j - \mu_j|^2 = \sigma^2_j$. Let $m = \lfloor{8\log(1/\delta)  \wedge T/2} \rfloor$ and $t = \lfloor T/m \rfloor$. Let 
        \[
        \bar{\mu}^1_j = \frac{1}{t}\sum_{s = 1}^{t} X_{j, s} ~, ~ \bar{\mu}^2_j = \frac{1}{t}\sum_{s = t+1}^{2t}X_{j, s} ~,  \ldots , ~ \bar{\mu}^m_j= \frac{1}{t}\sum_{s = (m-1)t+1}^{kt}X_{j, s},
        \]
        be $m$ empirical mean estimates, each one computed on $t$ data points. Let $\hat{\mu}_j$  be the median of these $m$ empirical means. Then with probability at least $1-\delta$ 
        \[ 
        |\hat{\mu}_j - \mu| \leq \sigma\sqrt{\frac{96\log(1/\delta)}{T}}.  
        \]
    \end{lemma}
    \begin{proof}
        By Chebyshev's inequality $\forall  \ell \in [m]$
        \[
        \P[|\bar{\mu}^{\ell}_j - \mu_j| \leq \sigma\sqrt{12/t}] \geq 3/4
        \]
        
        Let $\epsilon > 0$ and $Y_{\ell} = \ind{|\hat{\mu}^{\ell}_j  - \mu_j|>  \epsilon}$ for $\ell \in [k]$, For $\epsilon =  \sigma\sqrt{12/t}$,   $Y_{\ell} $ is stochastically dominated by a Bernoulli distribution with parameter  $p = 1/4$. Thus using Hoeffding's inequality for the tail of a binomial distribution we get
        \begin{align*} \P(|\hat{\mu}_j - \mu_j| > \epsilon) &= \P\left(\sum_{\ell = 1}^{m} Y_{\ell} \geq m/2 \right) \leq \P(Bin(m, 1/4) \geq m/2)\leq \exp(-2m(p - 1/2)^2) = \exp(-m/8)\\
        &= \delta.\end{align*}
    \end{proof}

\begin{claim}\label{clm:var_low}
    $\sigma^2_{j} = \E[|X_j-\mu_j|^2] \leq n/4 ~~ \forall j \in [k]$
\end{claim}
\begin{proof}
    \begin{align*}
        \E[|X_j-\mu_j|^2]  &= \E\left[\left| \sum_{i=1}^n X_{i,j}-\mu_{i,j} \right|^2\right] \\
        &\leq \sum_{i=1}^n\E[|X_{i,j}-\mu_{i,j}|^2] \\
        &\leq \sum_{i=1}^n 1/4 = n/4
    \end{align*}
    where the second line follows from the independence across users $i$ and triangle inequality. While the last line follows from Popoviciu's variance inequality.
\end{proof}

\section{Proofs about the Formulation 1 regret lower bounds (Section~\ref{sec:lower})}\label{app:lower_bound}
\nsqrtT*
\begin{proof}
Our proof is based on worst-case instances $\nu_{\vec{b}}$ defined for any vector $\vec{b} \in \{0,1\}^n$.
For each user $i \in [n]$, their reward distributions for the two arms are Bernoulli with means $\vec{\mu}_i = \left(\mu_{i,0}, \mu_{i,1}\right)$ where
\begin{equation}\vec{\mu}_i = \begin{cases} \left(\frac{1}{2} + \epsilon, \frac{1}{2}\right) &\text{if  } b_i = 0\\
\left(\frac{1}{2}, \frac{1}{2} + \epsilon\right) & \text{if  }b_i = 1\end{cases}\label{eq:app_WC_instance}\end{equation} where $\epsilon = \sqrt{\frac{1}{8T}}.$
We will lower bound the expected regret $\E_{\vec{b}}\left[R_T\left(\pi, \nu_{\vec{b}}\right)\right]$ over both the randomness of the draw of the vector $\vec{b} \sim \text{Unif}\left(\{0,1\}^n\right)$ and the distribution over outcomes $\P_{\pi\nu_{\vec{b}}}$. This will imply that for any policy $\pi$, there exists an instance $\nu_{\vec{b}}$ such that  \[R_T\left(\pi, \nu_{\vec{b}}\right) \geq \sqrt{\frac{T}{8}}\left(\frac{n}{8e} - \gamma \left(\frac{n}{8e} + \sqrt{\frac{n}{2\pi}}\right) \right).\]

Given an instance $\nu_{\vec{b}}$,  the following distributions $\vec{p}_1, \dots, \vec{p}_n$ with \[\vec{p}_i = \begin{cases} \left(1 - \frac{\gamma\norm{\vec{b}}_1}{n}, \frac{\gamma\norm{\vec{b}}_1}{n}\right) & \text{if } b_i = 0\\
    \left(\frac{\gamma (n - \norm{\vec{b}}_1)}{n}, 1 - \frac{\gamma (n - \norm{\vec{b}}_1)}{n} \right) & \text{if } b_i = 1.\end{cases}\] are feasible policies. This is because
    $n - \norm{\vec{b}}_1$ is the number of 0's in $\vec{b}$ and $\norm{\vec{b}}_1$ is the number of 1's in $\vec{b}$, so for any $i$ such that $b_i = 0$,  \[1 - \frac{\gamma\norm{\vec{b}}_1}{n} \geq \frac{\gamma}{n}\left(\sum_{i : b_i = 0} \left(1 - \frac{\gamma\norm{\vec{b}}_1}{n}\right) + \sum_{i : b_i= 1} \frac{\gamma\left(n - \norm{\vec{b}}_1\right)}{n}\right) = \gamma - \frac{\gamma \norm{\vec{b}}_1}{n}\] and \[\frac{\gamma\norm{\vec{b}}_1}{n} = \frac{\gamma}{n}\left(\sum_{i : b_i = 0} \frac{\gamma\norm{\vec{b}}_1}{n} + \sum_{i : b_i= 1} \left( 1- \frac{\gamma\left(n - \norm{\vec{b}}_1\right)}{n}\right)\right).\] Similarly, for any $i$ such that $b_i = 1$, \[\frac{\gamma (n - \norm{\vec{b}}_1)}{n} = \frac{\gamma}{n}\left(\sum_{i : b_i = 0} \left(1 - \frac{\gamma\norm{\vec{b}}_1}{n}\right) + \sum_{i : b_i= 1} \frac{\gamma\left(n - \norm{\vec{b}}_1\right)}{n}\right)\] and \[1 - \frac{\gamma (n - \norm{\vec{b}}_1)}{n} = 1 - \gamma + \frac{\gamma \norm{\vec{b}}_1}{n} \geq \frac{\gamma}{n}\left(\sum_{i : b_i = 0} \frac{\gamma\norm{\vec{b}}_1}{n} + \sum_{i : b_i= 1} \left( 1- \frac{\gamma\left(n - \norm{\vec{b}}_1\right)}{n}\right)\right) = \frac{\gamma \norm{\vec{b}}_1}{n}.\]
    
    After simplifying, this policy has an objective value of \begin{align*}&\sum_{i : b_i = 0} \left( \frac{1}{2} + \left(1 - \frac{\gamma\norm{\vec{b}}_1}{n}\right)\epsilon \right) + \sum_{i : b_i = 1} \left(\frac{1}{2} + \left(1 - \frac{\gamma (n - \norm{\vec{b}}_1)}{n} \right)\epsilon\right)\\
    = \, &\frac{n}{2} + \left(n - \norm{\vec{b}}_1\right)\left(1 - \frac{\gamma\norm{\vec{b}}_1}{n}\right)\epsilon + \norm{\vec{b}}_1 \left(1 - \frac{\gamma (n - \norm{\vec{b}}_1)}{n} \right)\epsilon\\
    = \, & \frac{n}{2} + \epsilon\left( n - \frac{1}{n} \cdot 2\gamma \norm{\vec{b}}_1\left(n - \norm{\vec{b}}_1\right)\right).\end{align*}
    Thus, the optimal policy's expected cumulative reward is at least \begin{equation}\frac{nT}{2} + \epsilon \left( nT - \frac{1}{n} \cdot 2T\gamma \norm{\vec{b}}_1\left(n - \norm{\vec{b}}_1\right)\right).\label{eq:app_opt_bound}\end{equation}

    Meanwhile, for any policy $\pi$, let $\pi_{i,0}^{(t)} = \pi_i\left(0 \mid \vec{h}_{t-1}\right)$ denote the probability that the policy chooses arm $0$ for user $i$ on round $t$ given the history $\vec{h}_{t-1}$. The value $\pi_{i,0}^{(t)}$ is therefore a random variable that depends on the history $\vec{h}_{t-1}$. Similarly, let $\pi_{i,1}^{(t)} = \pi_i\left(1 \mid \vec{h}_{t-1}\right)$.
    The expected cumulative reward of policy $\pi$ is \begin{align}&\E_{\pi\nu_{\vec{b}}}\left[ \sum_{t = 1}^T\left(\sum_{i : b_i = 0}\left(\left(\frac{1}{2} + \epsilon\right) \pi_{i,0}^{(t)} + \frac{1}{2} \pi_{i,1}^{(t)}\right) + \sum_{i : b_i = 1} \left(\frac{1}{2} \pi_{i,0}^{(t)} + \left(\frac{1}{2} + \epsilon\right)\pi_{i,1}^{(t)}\right)\right)\right]\nonumber\\
    =\,&\frac{nT}{2} + \epsilon \left(\sum_{t = 1}^T\left(\sum_{i : b_i = 0}\E_{\pi\nu_{\vec{b}}}\left[\pi_{i,0}^{(t)}\right] + \sum_{i : b_i = 1} \E_{\pi\nu_{\vec{b}}}\left[\pi_{i,1}^{(t)}\right]\right)\right)\nonumber\\
    =\, &\frac{nT}{2} + \epsilon \left(nT - \sum_{t = 1}^T\left(\sum_{i : b_i = 0}\E_{\pi\nu_{\vec{b}}}\left[\pi_{i,1}^{(t)}\right] + \sum_{i : b_i = 1} \E_{\pi\nu_{\vec{b}}}\left[\pi_{i,0}^{(t)}\right]\right)\right).\label{eq:app_policy_reward}\end{align}

    Therefore, the expected regret of $\pi$ is at least Equation~\eqref{eq:app_opt_bound} minus Equation~\eqref{eq:app_policy_reward}, which is \begin{equation}\epsilon \left(\sum_{t = 1}^T\left(\sum_{i : b_i = 0}\E_{\pi\nu_{\vec{b}}}\left[\pi_{i,1}^{(t)}\right] + \sum_{i : b_i = 1} \E_{\pi\nu_{\vec{b}}}\left[\pi_{i,0}^{(t)}\right]\right) - \frac{1}{n} \cdot 2T\gamma \norm{\vec{b}}_1\left(n - \norm{\vec{b}}_1\right)\right).\label{eq:app_regret_bound}\end{equation}

    We will begin incorporating the constraints by rewriting the first part of Equation~\eqref{eq:app_regret_bound} as \begin{align}&\sum_{t = 1}^T\left(\sum_{i : b_i = 0}\E_{\pi\nu_{\vec{b}}}\left[\pi_{i,1}^{(t)}\right] + \sum_{i : b_i = 1} \E_{\pi\nu_{\vec{b}}}\left[\pi_{i,0}^{(t)}\right]\right)\nonumber\\
    =\, & \sum_{t = 1}^T\left(\sum_{i : b_i = 0}\E_{\pi\nu_{\vec{b}}}\left[\pi_{i,1}^{(t)} - \frac{\gamma}{n} \sum_{j = 1}^n \pi_{j,1}^{(t)}\right] + \sum_{i : b_i = 1} \E_{\pi\nu_{\vec{b}}}\left[\pi_{i,0}^{(t)}- \frac{\gamma}{n} \sum_{j = 1}^n \pi_{j,0}^{(t)}\right]\right)\label{eq:app_lb_add_constraints}\\
    &+ \sum_{t = 1}^T\left(\frac{\gamma\left(n - \norm{\vec{b}}_1\right)}{n}\sum_{j = 1}^n \E_{\pi\nu_{\vec{b}}}\left[\pi_{j,1}^{(t)}\right] + \frac{\gamma\norm{\vec{b}}_1}{n}\sum_{j = 1}^n \E_{\pi\nu_{\vec{b}}}\left[\pi_{j,0}^{(t)}\right]\right).\nonumber\end{align}
    If $\norm{\vec{b}}_1 < \frac{n}{2}$, then $ \frac{\gamma\left(n - \norm{\vec{b}}_1\right)}{n} >  \frac{\gamma\norm{\vec{b}}_1}{n}$, so \begin{align*}
       \sum_{t = 1}^T\left(\frac{\gamma\left(n - \norm{\vec{b}}_1\right)}{n}\sum_{j = 1}^n \E_{\pi\nu_{\vec{b}}}\left[\pi_{j,1}^{(t)}\right] + \frac{\gamma\norm{\vec{b}}_1}{n}\sum_{j = 1}^n \E_{\pi\nu_{\vec{b}}}\left[\pi_{j,0}^{(t)}\right]\right) &\geq  \frac{\gamma\norm{\vec{b}}_1}{n}\E_{\pi\nu_{\vec{b}}}\left[\sum_{t = 1}^T\sum_{j = 1}^n \left(\pi_{j,1}^{(t)} + \pi_{j,0}^{(t)}\right)\right] \\
        &= \gamma \norm{\vec{b}}_1 T.\end{align*}
    Similarly, if $\norm{\vec{b}}_1 > \frac{n}{2}$, \begin{align*}
       \sum_{t = 1}^T\left(\frac{\gamma\left(n - \norm{\vec{b}}_1\right)}{n}\sum_{j = 1}^n \E_{\pi\nu_{\vec{b}}}\left[\pi_{j,1}^{(t)}\right] + \frac{\gamma\norm{\vec{b}}_1}{n}\sum_{j = 1}^n \E_{\pi\nu_{\vec{b}}}\left[\pi_{j,0}^{(t)}\right]\right) &\geq  \frac{\gamma\left(n - \norm{\vec{b}}_1\right)}{n}\E_{\pi\nu_{\vec{b}}}\left[\sum_{t = 1}^T\sum_{j = 1}^n \left(\pi_{j,1}^{(t)} + \pi_{j,0}^{(t)}\right)\right] \\
        &= \gamma \left(n - \norm{\vec{b}}_1\right) T.
    \end{align*}

    Therefore, \begin{align}&\E_{\vec{b} \sim \{0,1\}^n}\left[\sum_{t = 1}^T\left(\frac{\gamma\left(n - \norm{\vec{b}}_1\right)}{n}\sum_{j = 1}^n \E_{\pi\nu_{\vec{b}}}\left[\pi_{j,1}^{(t)}\right] + \frac{\gamma\norm{\vec{b}}_1}{n}\sum_{j = 1}^n \E_{\pi\nu_{\vec{b}}}\left[\pi_{j,0}^{(t)}\right]\right)\right]\label{eq:app_first_terms}\\
    \geq\, &\gamma T\E_{\vec{b}}\left[\norm{\vec{b}}_1\textbf{1}_{\left\{\norm{\vec{b}}_1 < \frac{n}{2}\right\}} + \left(n - \norm{\vec{b}}_1\right)\textbf{1}_{\left\{\norm{\vec{b}}_1 > \frac{n}{2}\right\}}\right]\nonumber\\
    =\, &\gamma T\E_{\vec{b}}\left[\norm{\vec{b}}_1\left(1 - \textbf{1}_{\left\{\norm{\vec{b}}_1 > \frac{n}{2}\right\}}\right) + \left(n - \norm{\vec{b}}_1\right)\textbf{1}_{\left\{\norm{\vec{b}}_1 > \frac{n}{2}\right\}}\right]\nonumber\\
    =\, &\gamma T\left(\E_{\vec{b}}\left[\norm{\vec{b}}_1\right] + n\E\left[\textbf{1}_{\left\{\norm{\vec{b}}_1 > \frac{n}{2}\right\}}\right] - 2\E\left[\norm{\vec{b}}_1\textbf{1}_{\left\{\norm{\vec{b}}_1 > \frac{n}{2}\right\}}\right]\right).\label{eq:app_bin}\end{align}
    When $\vec{b}\sim \text{Unif}\left(\{0,1\}^n\right)$, $\norm{\vec{b}}_1 \sim \text{Bin}\left(n, \frac{1}{2}\right)$. Therefore, Equation~\eqref{eq:app_bin} is equal to
    \[\gamma T\left(n - 2\E\left[\norm{\vec{b}}_1 \, \left| \, \norm{\vec{b}}_1 > \frac{n}{2}\right.\right]\Pr\left[\norm{\vec{b}}_1 > \frac{n}{2}\right]\right) = \gamma T\left(n - \E\left[\norm{\vec{b}}_1 \, \left| \, \norm{\vec{b}}_1 > \frac{n}{2}\right.\right]\right).\] Since $\norm{\vec{b}}_1 \sim \text{Bin}\left(n, \frac{1}{2}\right)$, \[\E\left[\norm{\vec{b}}_1 \, \left| \, \norm{\vec{b}}_1 > \frac{n}{2}\right.\right] = \frac{n}{2^n}\left(2^{n-1} + {n-1 \choose \frac{n-1}{2}}\right)\] \citep{stackexchange} and by Stirling's approximation, \[\E\left[\norm{\vec{b}}_1 \, \left| \, \norm{\vec{b}}_1 > \frac{n}{2}\right.\right] < \frac{n}{2} + \sqrt{\frac{n}{2\pi}}.\] We can use these facts to bound Equation~\eqref{eq:app_first_terms} as follows: \begin{equation}\E_{\vec{b} \sim \{0,1\}^n}\left[\sum_{t = 1}^T\left(\frac{\gamma\left(n - \norm{\vec{b}}_1\right)}{n}\sum_{j = 1}^n \E_{\pi\nu_{\vec{b}}}\left[\pi_{j,1}^{(t)}\right] + \frac{\gamma\norm{\vec{b}}_1}{n}\sum_{j = 1}^n \E_{\pi\nu_{\vec{b}}}\left[\pi_{j,0}^{(t)}\right]\right)\right] \geq \gamma T \left( \frac{n}{2} - \sqrt{\frac{n}{2\pi}}\right).\end{equation}

Returning to Equation~\eqref{eq:app_lb_add_constraints}, we will next bound \begin{align*}&\E_{\vec{b}}\left[\sum_{t = 1}^T\left(\sum_{i : b_i = 0}\E_{\pi\nu_{\vec{b}}}\left[\pi_{i,1}^{(t)} - \frac{\gamma}{n} \sum_{j = 1}^n \pi_{j,1}^{(t)}\right] + \sum_{i : b_i = 1} \E_{\pi\nu_{\vec{b}}}\left[\pi_{i,0}^{(t)}- \frac{\gamma}{n} \sum_{j = 1}^n \pi_{j,0}^{(t)}\right]\right)\right]\\
= \, &\sum_{i = 1}^n \E_{\vec{b}}\left[\sum_{t = 1}^T\left(\E_{\pi\nu_{\vec{b}}}\left[\pi_{i,1}^{(t)} - \frac{\gamma}{n} \sum_{j = 1}^n \pi_{j,1}^{(t)}\right]\textbf{1}_{\{b_i = 0\}}  + \E_{\pi\nu_{\vec{b}}}\left[\pi_{i,0}^{(t)} - \frac{\gamma}{n} \sum_{j = 1}^n \pi_{j,0}^{(t)}\right]\textbf{1}_{\{b_i = 1\}}\right)\right].
\end{align*}
    For each user $i \in [n]$, we will therefore lower bound \begin{align}&\E_{\vec{b}}\left[\sum_{t = 1}^T\left(\E_{\pi\nu_{\vec{b}}}\left[\pi_{i,1}^{(t)} - \frac{\gamma}{n} \sum_{j = 1}^n \pi_{j,1}^{(t)}\right]\textbf{1}_{\{b_i = 0\}}  + \E_{\pi\nu_{\vec{b}}}\left[\pi_{i,0}^{(t)} - \frac{\gamma}{n} \sum_{j = 1}^n \pi_{j,0}^{(t)}\right]\textbf{1}_{\{b_i = 1\}}\right)\right]\nonumber\\
    =\, & \frac{1}{2} \left(\E_{\vec{b}}\left[\left.\E_{\pi\nu_{\vec{b}}}\left[\sum_{t = 1}^T\left(\pi_{i,1}^{(t)} - \frac{\gamma}{n} \sum_{j = 1}^n \pi_{j,1}^{(t)}\right)\right] \, \right| \, b_i = 0\right]  + \E_{\vec{b}}\left[\left.\E_{\pi\nu_{\vec{b}}}\left[\sum_{t = 1}^T\left(\pi_{i,0}^{(t)} - \frac{\gamma}{n} \sum_{j = 1}^n \pi_{j,0}^{(t)}\right)\right] \, \right| \, b_i = 1\right]\right).\label{eq:app_isolate_agent}\end{align} Let $\vec{b}_{-i} \in \{0,1\}^{n-1}$ denote the vector $\vec{b}$ with all components except the $i^{th}$ component. Moreover, to simplify notation, let $\P_{i,0}$ denote the distribution over outcomes $(\vec{A}_1, \vec{X}_1, \dots, \vec{A}_T, \vec{X}_T) \in (\{0,1\}^n \times \{0,1\}^n)^T$ defined by first drawing $\vec{b}_{-i}\sim \text{Unif}\left(\{0,1\}^{n-1}\right)$ and then running the policy $\pi$ on the instance $\nu_{(0, \vec{b}_{-i})}$.
Similarly, let $\P_{i,1}$ denote the distribution over outcomes $(\vec{A}_1, \vec{X}_1, \dots, \vec{A}_T, \vec{X}_T) \in (\{0,1\}^n \times \{0,1\}^n)^T$ defined by first drawing $\vec{b}_{-i}\sim \text{Unif}\left(\{0,1\}^{n-1}\right)$ and then running the policy $\pi$ on the instance $\nu_{(1, \vec{b}_{-i})}$. We can then rewrite Equation~\eqref{eq:app_isolate_agent} as \begin{align}&\frac{1}{2} \left(\E_{i,0}\left[\sum_{t = 1}^T\left(\pi_{i,1}^{(t)} - \frac{\gamma}{n} \sum_{j = 1}^n \pi_{j,1}^{(t)}\right)\right]  + \E_{i,1}\left[\sum_{t = 1}^T\left(\pi_{i,0}^{(t)} - \frac{\gamma}{n} \sum_{j = 1}^n \pi_{j,0}^{(t)}\right)\right] \right)\nonumber\\
=\, & \frac{1}{2} \left(\E_{i,0}\left[\sum_{t = 1}^T\left(\pi_{i,1}^{(t)} - \frac{\gamma}{n} \sum_{j = 1}^n \pi_{j,1}^{(t)}\right)\right]  + \E_{i,1}\left[(1-\gamma)T - \sum_{t = 1}^T\left(\pi_{i,1}^{(t)} - \frac{\gamma}{n} \sum_{j = 1}^n \pi_{j,1}^{(t)}\right)\right]\right).
\label{eq:app_simplified_notation}\end{align}

Based on the constraints, we know that $\pi_{i,1}^{(t)} - \frac{\gamma}{n} \sum_{j = 1}^n \pi_{j,1}^{(t)} \geq 0$ and $\pi_{i,0}^{(t)} - \frac{\gamma}{n} \sum_{j = 1}^n \pi_{j,0}^{(t)} \geq 0$ with probability 1. Therefore, by Markov's inequality and the Bretagnolle–Huber inequality, \begin{align*}
    &\frac{1}{2} \left(\E_{i,0}\left[\sum_{t = 1}^T\left(\pi_{i,1}^{(t)} - \frac{\gamma}{n} \sum_{j = 1}^n \pi_{j,1}^{(t)}\right)\right]  + \E_{i,1}\left[(1-\gamma)T - \sum_{t = 1}^T\left(\pi_{i,1}^{(t)} - \frac{\gamma}{n} \sum_{j = 1}^n \pi_{j,1}^{(t)}\right)\right]\right)\\
    \geq \, & \frac{T(1- \gamma)}{4} \left(\P_{i,0}\left[\sum_{t = 1}^T\left(\pi_{i,1}^{(t)} - \frac{\gamma}{n} \sum_{j = 1}^n \pi_{j,1}^{(t)}\right) \geq \frac{T(1-\gamma)}{2}\right]\right.\\
    &+ \left.\P_{i,1}\left[\sum_{t = 1}^T\left(\pi_{i,1}^{(t)} - \frac{\gamma}{n} \sum_{j = 1}^n \pi_{j,1}^{(t)}\right) < \frac{T(1-\gamma)}{2}\right]\right)\\
    \geq \, & \frac{T(1- \gamma)}{8}\exp\left( - D\left(\P_{i,0}, \P_{i,1}\right)\right).
\end{align*}

In the following claim, we bound $D\left(\P_{i,0}, \P_{i,1}\right).$

    \begin{claim}\label{claim:KL}
    $D\left(\P_{i,0}, \P_{i,1}\right) \leq 8\epsilon^2T.$
\end{claim}

\begin{proof}[Proof of Claim~\ref{claim:KL}]
In this proof, we will use the following notation to distinguish the reward distributions for each instance $\nu_{\vec{b}}$. For any vector of rewards $\vec{x}_t = \left(x_{t,1}, \dots, x_{t,n}\right) \in \{0,1\}^n$ and any choice of arms $\vec{a}_t = \left(a_{t,1}, \dots, a_{t,n}\right) \in \{0,1\}^n$, we will use the notation $f^{\vec{b}}_{\vec{a}_t}(\vec{x}_t)$ to denote the probability that the platform receives rewards $\vec{x}_t$ under instance $\nu_{\vec{b}}$  after choosing arms $\vec{a}_t.$ We also use $f_{i,0}^{(b_i)} : \{0,1\} \to [0,1]$ to denote the PMF of arm $0$ for user $i$ and $f_{i,1}^{(b_i)} : \{0,1\} \to [0,1]$ to denote the PMF of arm $1$ for user $i$.
In other words, $f_{i, 0}^{(0)}$ is the Bern$\left( \frac{1}{2} + \epsilon\right)$ PMF, $f_{i, 1}^{(0)}$ is the Bern$\left( \frac{1}{2}\right)$ PMF, $f_{i, 0}^{(1)}$ is the Bern$\left( \frac{1}{2}\right)$ PMF, and $f_{i,1}^{(1)}$ is the Bern$\left( \frac{1}{2} + \epsilon\right)$ PMF. With this notation, \begin{equation}f^{\vec{b}}_{\vec{a}_t}(\vec{x}_t) = \prod_{i = 1}^n f^{(b_i)}_{i,a_{t,i}}\left(x_{t,i}\right).\label{eq:app_product_dist}\end{equation}

Moving now to KL divergence between $\P_{i,0}$ and $\P_{i,1}$, let $f_{i,0} : (\{0,1\}^n \times \{0,1\}^n)^T \to [0,1]$ be the probability mass function of the distribution $\P_{i,0}$, and define $f_{i,1}$ similarly.
By definition, \begin{equation}D\left(\P_{i,0}, \P_{i,1}\right) = \sum_{\left(\vec{a}_t, \vec{x}_t\right)_{t = 1}^T} f_{i,0}\left(\left(\vec{a}_t, \vec{x}_t\right)_{t = 1}^T\right) \log \frac{f_{i,0}\left(\left(\vec{a}_t, \vec{x}_t\right)_{t = 1}^T\right)}{f_{i,1}\left(\left(\vec{a}_t, \vec{x}_t\right)_{t = 1}^T\right)}.\label{eq:app_complex_KL}\end{equation}

We will begin by simplifying the logarithm in Equation~\eqref{eq:app_complex_KL}. Beginning with the numerator of
the logarithm, we have that \begin{align*}f_{i,0}\left(\left(\vec{a}_t, \vec{x}_t\right)_{t = 1}^T\right) &= \P_{i,0}\left[\left(\vec{A}_t, \vec{X}_t\right)_{t = 1}^T = \left(\vec{a}_t, \vec{x}_t\right)_{t = 1}^T\right]\\
&= \frac{1}{2^{n-1}}\sum_{\vec{b}_{-i} \in \{0,1\}^{n-1}} \P_{\pi\nu_{\left(0, \vec{b}_{-i}\right)}}\left[\left(\vec{A}_t, \vec{X}_t\right)_{t = 1}^T = \left(\vec{a}_t, \vec{x}_t\right)_{t = 1}^T\right].\end{align*}
Using the notation defined in Section~\ref{sec:notation} (Equation~\eqref{eq:PDF}), we have that \begin{align*}f_{i,0}\left(\left(\vec{a}_t, \vec{x}_t\right)_{t = 1}^T\right) &= \frac{1}{^{n-1}}\sum_{\vec{b}_{-i}} f_{\pi \nu_{\left(0, \vec{b}_{-i}\right)}}\left(\left(\vec{a}_t, \vec{x}_t\right)_{t = 1}^T\right)\\
&= \frac{1}{2^{n-1}}\sum_{\vec{b}_{-i}} \prod_{t = 1}^T \pi(\vec{a}_t \mid \vec{a}_1, \vec{x}_1, \dots, \vec{a}_{t-1}, \vec{x}_{t-1})f^{\left(0, \vec{b}_{-i}\right)}_{\vec{a}_t}(\vec{x}_t).\end{align*}
Applying Equation~\eqref{eq:app_product_dist}, we have that \[f_{i,0}\left(\left(\vec{a}_t, \vec{x}_t\right)_{t = 1}^T\right) = \frac{1}{2^{n-1}}\sum_{\vec{b}_{-i}} \prod_{t = 1}^T\left( \pi(\vec{a}_t \mid \vec{a}_1, \vec{x}_1, \dots, \vec{a}_{t-1}, \vec{x}_{t-1})f^{(0)}_{i,a_{t,i}}\left(x_{t,i}\right)\prod_{j \not= i}f^{(b_{j})}_{j,a_{t,j}}(x_{t,j})\right)\] where $b_j$ indicates the $j^{th}$ component of the vector $\vec{b}_{-i}.$
Rearranging the product within the summation, we have that $f_{i,0}\left(\left(\vec{a}_t, \vec{x}_t\right)_{t = 1}^T\right)$ is equal to \[\frac{1}{2^{n-1}}\sum_{\vec{b}_{-i}} \left(\prod_{t = 1}^T\left( \pi(\vec{a}_t \mid \vec{a}_1, \vec{x}_1, \dots, \vec{a}_{t-1}, \vec{x}_{t-1})f^{(0)}_{i,a_{t,i}}\left(x_{t,i}\right)\right)\prod_{t = 1}^T\left(\prod_{j \not= i}f^{(b_{j})}_{j,a_{t,j}}(x_{t,j})\right)\right).\] Since $\prod_{t = 1}^T \pi(\vec{a}_t \mid \vec{a}_1, \vec{x}_1, \dots, \vec{a}_{t-1}, \vec{x}_{t-1})f^{(0)}_{i,a_{t,i}}\left(x_{t,i}\right)$ does not depend on $\vec{b}_{-i}$, we rearrange the summation over $\vec{b}_{-i}$ as \begin{equation}f_{i,0}\left(\left(\vec{a}_t, \vec{x}_t\right)_{t = 1}^T\right) = \frac{1}{2^{n-1}}\prod_{t = 1}^T\left( \pi(\vec{a}_t \mid \vec{a}_1, \vec{x}_1, \dots, \vec{a}_{t-1}, \vec{x}_{t-1})f^{(0)}_{i,a_{t,i}}\left(x_{t,i}\right)\right)\sum_{\vec{b}_{-i}} \prod_{t = 1}^T\prod_{j \not= i}f^{(b_{j})}_{j,a_{t,j}}(x_{t,j}).\label{eq:app_S0}\end{equation} Similarly, \begin{equation}f_{i,1}\left(\left(\vec{a}_t, \vec{x}_t\right)_{t = 1}^T\right) = \frac{1}{2^{n-1}}\prod_{t = 1}^T\left( \pi(\vec{a}_t \mid \vec{a}_1, \vec{x}_1, \dots, \vec{a}_{t-1}, \vec{x}_{t-1})f^{(1)}_{i,a_{t,i}}\left(x_{t,i}\right)\right)\sum_{\vec{b}_{-i}} \prod_{t = 1}^T\prod_{j \not= i}f^{(b_{j})}_{j,a_{t,j}}(x_{t,j}).\label{eq:app_S1}\end{equation}

We now return to the logarithm in Equation~\eqref{eq:app_complex_KL}. Based on Equations~\eqref{eq:app_S0} and \eqref{eq:app_S1}, much of the numerator and denominator cancel out, leaving us with \[\log \frac{f_{i,0}\left(\left(\vec{a}_t, \vec{x}_t\right)_{t = 1}^T\right)}{f_{i,1}\left(\left(\vec{a}_t, \vec{x}_t\right)_{t = 1}^T\right)} = \log \prod_{t = 1}^T \frac{f^{(0)}_{ i,a_{t,i}}\left(x_{t,i}\right)}{f^{(1)}_{i, a_{t,i}}\left(x_{t,i}\right)} = \sum_{t = 1}^T \log \frac{f^{(0)}_{i, a_{t,i}}\left(x_{t,i}\right)}{f^{(1)}_{i, a_{t,i}}\left(x_{t,i}\right)}.\]

We can therefore can write the KL divergence as \[D\left(\P_{i,0}, \P_{i,1}\right) = \sum_{t = 1}^T\E_{i,0}\left[ \log \frac{f^{(0)}_{i,A_{t,i}}\left(X_{t,i}\right)}{f^{(1)}_{i,A_{t,i}}\left(X_{t,i}\right)}\right].\] Moreover, by the law of total expectation, $D\left(\P_{i,0}, \P_{i,1}\right)$ is equal to \begin{equation}\sum_{t = 1}^T\left(\E_{i,0}\left[\left. \log \frac{f^{(0)}_{i,0}\left(X_{t,i}\right)}{f^{(1)}_{i,0}\left(X_{t,i}\right)}\, \right| \, A_{t,i} = 0\right]\P\left[A_{t,i} = 0\right] + \E_{i,0}\left[\left. \log \frac{f^{(0)}_{i,1}\left(X_{t,i}\right)}{f^{(1)}_{i,1}\left(X_{t,i}\right)}\, \right| \, A_{t,i} = 1\right]\P\left[A_{t,i} = 1\right]\right).\label{eq:app_LOTE}\end{equation}

Inspecting each conditional expectation in this sum, we have that \begin{align*}&\E_{i,0}\left[\left. \log \frac{f^{(0)}_{i,0}\left(X_{t,i}\right)}{f^{(1)}_{i,0}\left(X_{t,i}\right)}\, \right| \, A_{t,i} = 0\right]\\
=\, &\P_{i,0}\left[X_{t,i} = 0 \mid A_{t,i} = 0\right]\cdot \log \frac{f^{(0)}_{i,0}\left(0\right)}{f^{(1)}_{i,0}\left(0\right)} + \P_{i,0}\left[X_{t,i} = 1 \mid A_{t,i} = 0\right]\cdot \log \frac{f^{(0)}_{i,0}\left(1\right)}{f^{(1)}_{i,0}\left(1\right)}.\end{align*}
By Equation~\eqref{eq:app_WC_instance}, for any instance $\nu_{\vec{b}}$ such that $b_i = 0$, the probability that $X_{t,i} = 0$ given that $A_{t,i} = 0$ is $\frac{1}{2} - \epsilon = f^{(0)}_{i,0}\left(0\right)$. Similarly, the probability that $X_{t,i} = 1$ given that $A_{t,i} = 0$ is $\frac{1}{2} + \epsilon = f^{(0)}_{i,0}\left(1\right)$. Therefore, \begin{equation}\E_{i,0}\left[\left. \log \frac{f^{(0)}_{i,0}\left(X_{t,i}\right)}{f^{(1)}_{i,0}\left(X_{t,i}\right)}\, \right| \, A_{t,i} = 0\right] =f^{(0)}_{i,0}\left(0\right) \log \frac{f^{(0)}_{i,0}\left(0\right)}{f^{(1)}_{i,0}\left(0\right)} + f^{(0)}_{i,0}\left(1\right)\cdot \log \frac{f^{(0)}_{i,0}\left(1\right)}{f^{(1)}_{i,0}\left(1\right)} = D\left(f^{(0)}_{i,0}, f^{(1)}_{i,0}\right).\label{eq:app_fi0}\end{equation} Similarly, \begin{equation}\E_{i,0}\left[\left. \log \frac{f^{(0)}_{i,1}\left(X_{t,i}\right)}{f^{(1)}_{i,1}\left(X_{t,i}\right)}\, \right| \, A_{t,i} = 1\right] =D\left(f^{(0)}_{i,1}, f^{(1)}_{i,1}\right).\label{eq:app_fi1}\end{equation}

Let $N_{i,0}(T)$ be the number of rounds that user $i$ is shown arm $0$ and let $N_{i,1}(T)$ be the number of rounds that user $i$ is shown arm $0$, so $N_{i,0}(T) + N_{i,1}(T) = T$. Combining Equations~\eqref{eq:app_LOTE}, \eqref{eq:app_fi0}, and \eqref{eq:app_fi1}, we have that \begin{align*}D\left(\P_{i,0}, \P_{i,1}\right) &= \sum_{t = 1}^T \left(D\left(f^{(0)}_{i,0}, f^{(1)}_{i,0}\right)\P_{i,0}\left[A_{t,i} = 0\right] + D\left(f^{(0)}_{i,1}, f^{(1)}_{i,1}\right)\P_{i,0}\left[A_{t,i} = 1\right]\right)\\
&= D\left(f^{(0)}_{i,0}, f^{(1)}_{i,0}\right)\E_{i,0}\left[N_{i,0}(T)\right] + D\left(f^{(0)}_{i,1}, f^{(1)}_{i,1}\right)\E_{i,0}\left[N_{i,1}(T)\right].\end{align*}

Since $f_{i, 0}^{(0)}$ is Bern$\left( \frac{1}{2} + \epsilon\right)$, $f_{i, 1}^{(0)}$ is Bern$\left( \frac{1}{2}\right)$, $f_{i, 0}^{(1)}$ is Bern$\left( \frac{1}{2}\right)$, and $f_{i,1}^{(1)}$ is Bern$\left( \frac{1}{2} + \epsilon\right)$, \[D\left(\P_{i,0}, \P_{i,1}\right) \leq 8\epsilon^2\left(\E_{i,0}\left[N_{i,0}(T)\right] + \E_{i,0}\left[N_{i,1}(T)\right]\right) = 8\epsilon^2T.\]
\end{proof}

By Claim~\ref{claim:KL}, $D\left(\P_{i,0}, \P_{i,1}\right) \leq 8\epsilon^2T$, so if we set $\epsilon = \sqrt{\frac{1}{8T}}$, we have that \begin{equation}\frac{1}{2} \left(\E_{i,0}\left[\sum_{t = 1}^T\left(\pi_{i,1}^{(t)} - \frac{\gamma}{n} \sum_{j = 1}^n \pi_{j,1}^{(t)}\right)\right]  + \E_{i,1}\left[(1-\gamma)T - \sum_{t = 1}^T\left(\pi_{i,1}^{(t)} - \frac{\gamma}{n} \sum_{j = 1}^n \pi_{j,1}^{(t)}\right)\right]\right) \geq \frac{T(1 - \gamma)}{8e}.\label{eq:app_fix_eps}\end{equation}

Combining Equations~\eqref{eq:app_regret_bound}, \eqref{eq:app_lb_add_constraints}, and \eqref{eq:app_fix_eps}, we have that the regret is lower bounded by \[\sqrt{\frac{1}{8T}}\left(\frac{nT(1-\gamma)}{8e} + \gamma T\left(\frac{n}{2} - \sqrt{\frac{n}{2\pi}}\right) - \frac{2T\gamma}{n}\E_{\vec{b}\sim \{0,1\}^n}\left[\norm{\vec{b}}_1\left(n - \norm{\vec{b}}_1\right)\right]\right).\] Since $\E_{\vec{b}\sim \{0,1\}^n}\left[\norm{\vec{b}}_1\left(n - \norm{\vec{b}}_1\right)\right] = \frac{n}{4}(n-1)$, we have that the expected regret is lower bounded by \[\sqrt{\frac{T}{8}}\left(\frac{n(1-\gamma)}{8e} + \gamma \left(\frac{n}{2} - \sqrt{\frac{n}{2\pi}}\right) - \frac{\gamma(n-1)}{2}\right) \geq \sqrt{\frac{T}{8}}\left(\frac{n}{8e} - \gamma \left(\frac{n}{8e} + \sqrt{\frac{n}{2\pi}}\right) \right).\]

\end{proof}

\sqrtnT*
\begin{proof}
    We begin by defining the worst-case instance $\nu$ where for each user $i \in [n]$, their reward distributions for the $k$ arms are Bernoulli with means \[\vec{\mu}_1 = \cdots = \vec{\mu}_n = \left(\frac{1}{2}+ \epsilon, \frac{1}{2}, \frac{1}{2}, \dots, \frac{1}{2}\right)\] where $\epsilon = \sqrt{\frac{k-1}{8nT}}$. We will use the notation $N_{i,j}(T)$ to denote the number of rounds that user $i$ is shown arm $j$ and $N_j(T) = \sum_{i = 1}^n N_{i,j}(T)$ to denote the total number of rounds that all users are shown arm $j$. This means that $\sum_{j = 1}^k N_j(T) = nT.$
 Under instance $\nu$, the optimal policy obtains a reward of $nT\left(\frac{1}{2} + \epsilon\right)$. Meanwhile, an arbitrary policy $\pi$ will obtain a reward of \[\left(\frac{1}{2} + \epsilon\right)\E_{\pi\nu}\left[N_1(T)\right] + \frac{1}{2}\sum_{j = 2}^n \E_{\pi\nu}\left[N_j(T)\right] = \frac{nT}{2} + \epsilon \E_{\pi\nu}\left[N_1(T)\right].\]
    Therefore, the regret of policy $\pi$ on instance $\nu$ is \[R_T(\pi, \nu) = \epsilon\left(nT - \E_{\pi\nu}\left[N_1(T)\right]\right).\]
    
    Fix a policy $\pi$ and
    let $j^* = \argmin_{j > 2}\E_{\pi\nu}[N_j(T)].$ Since $\sum_{j = 2}^k N_j(T) \leq nT,$ we have that $\E_{\pi\nu}\left[N_{j^*}(T)\right] \leq \frac{nT}{k-1}$. We now use $j^*$ to construct a second worst-case instance $\nu'$ where for each user $i \in [n]$, \[\mu_{i,j} = \begin{cases}\frac{1}{2} + \epsilon &\text{if } j = 1\\
    \frac{1}{2} + 2\epsilon &\text{if } j = j^*\\
    \frac{1}{2} &\text{else.}\end{cases}\]

    Under instance $\nu'$, the optimal policy obtains a reward of $nT\left(\frac{1}{2} + 2\epsilon\right)$. Meanwhile, policy $\pi$ will obtain a reward of \[\left(\frac{1}{2} + \epsilon\right)\E_{\pi\nu'}\left[N_1(T)\right] + \left(\frac{1}{2} + 2\epsilon\right)\E_{\pi\nu'}\left[N_{j^*}(T)\right] + \frac{1}{2}\sum_{j \not\in \{1, j^*\}} \E_{\pi\nu'}\left[N_j(T)\right] = \frac{nT}{2} + \epsilon \E_{\pi\nu'}\left[N_1(T)\right]+ 2\epsilon \E_{\pi\nu'}\left[N_{j^*}(T)\right].\]
Therefore, \begin{align*}R_T(\pi, \nu') &= \epsilon\left(2nT - \E_{\pi\nu'}\left[N_1(T)\right]- 2\E_{\pi\nu'}\left[N_{j^*}(T)\right]\right)\\
&= \epsilon\left(2\sum_{j = 1}^k \E_{\pi\nu'}\left[N_j(T)\right] - \E_{\pi\nu'}\left[N_1(T)\right]- 2\E_{\pi\nu'}\left[N_{j^*}(T)\right]\right)\\
&\geq \epsilon \E_{\pi\nu'}\left[N_1(T)\right].\end{align*}

 By Markov's inequality, \[R_T(\pi, \nu) + R_T(\pi, \nu') \geq \frac{\epsilon nT}{2} \left(\P_{\pi \nu}\left[N_1(T) \leq \frac{nT}{2}\right] + \P_{\pi \nu'}\left[N_1(T) > \frac{nT}{2}\right]\right),\] so by the Bretagnolle–Huber inequality, \begin{equation}R_T(\pi, \nu) + R_T(\pi, \nu) \geq \frac{\epsilon nT}{4} \exp\left(-D\left(\P_{\pi \nu}, \P_{\pi \nu'}\right)\right).\label{eq:BH}\end{equation}

\begin{claim}\label{claim:second_KL}
For $\epsilon < \frac{1}{5}$, $D\left(\P_{\pi \nu}, \P_{\pi \nu'}\right) \leq \frac{4\epsilon^2nT}{k-1}.$
\end{claim}
\begin{proof}[Proof of Claim~\ref{claim:second_KL}]
In this proof, we will use the following notation to distinguish the reward distributions for the instances $\nu$ and $\nu'$. For any vector of rewards $\vec{x}_t = \left(x_{t,1}, \dots, x_{t,n}\right) \in \{0,1\}^n$ and any choice of arms $\vec{a}_t = \left(a_{t,1}, \dots, a_{t,n}\right) \in [k]^n$, we use the notation $f_{\vec{a}_t}(\vec{x}_t)$ (respectively, $f'_{\vec{a}_t}(\vec{x}_t)$) to denote the probability that the platform receives rewards $\vec{x}_t$ after choosing arms $\vec{a}_t$ under the instance $\nu$ (respectively, $\nu'$). We also use $f_{i,j} : \{0,1\} \to [0,1]$ (respectively, $f'_{i,j} : \{0,1\} \to [0,1]$) to denote the PMF of arm $j$ for user $i$.
With this notation, \begin{equation}f_{\vec{a}_t}(\vec{x}_t) = \prod_{i = 1}^n f_{i,a_{t,i}}\left(x_{t,i}\right)\label{eq:second_product_dist}\end{equation} and
\begin{equation}f'_{\vec{a}_t}(\vec{x}_t) = \prod_{i = 1}^n f'_{i,a_{t,i}}\left(x_{t,i}\right).\label{eq:second_prime_product_dist}\end{equation}

By definition, \begin{equation}D\left(\P_{\pi \nu}, \P_{\pi \nu'}\right) = \sum_{\left(\vec{a}_t, \vec{x}_t\right)_{t = 1}^T} f_{\pi\nu}\left(\left(\vec{a}_t, \vec{x}_t\right)_{t = 1}^T\right) \log \frac{f_{\pi\nu}\left(\left(\vec{a}_t, \vec{x}_t\right)_{t = 1}^T\right)}{f_{\pi\nu'}\left(\left(\vec{a}_t, \vec{x}_t\right)_{t = 1}^T\right)}.\label{eq:second_complex_KL}\end{equation}

We will begin by simplifying the logarithm in Equation~\eqref{eq:second_complex_KL}. Beginning with the numerator of
the logarithm and using the notation defined in Section~\ref{sec:notation} (Equation~\eqref{eq:PDF}), we have that \[f_{\pi\nu}\left(\left(\vec{a}_t, \vec{x}_t\right)_{t = 1}^T\right) = \P_{\pi\nu}\left[\left(\vec{A}_t, \vec{X}_t\right)_{t = 1}^T = \left(\vec{a}_t, \vec{x}_t\right)_{t = 1}^T\right] = \prod_{t = 1}^T \pi(\vec{a}_t \mid \vec{a}_1, \vec{x}_1, \dots, \vec{a}_{t-1}, \vec{x}_{t-1})f_{\vec{a}_t}(\vec{x}_t).\]
By Equation~\eqref{eq:second_product_dist}, we have that \begin{equation}f_{\pi\nu}\left(\left(\vec{a}_t, \vec{x}_t\right)_{t = 1}^T\right) = \prod_{t = 1}^T \pi(\vec{a}_t \mid \vec{a}_1, \vec{x}_1, \dots, \vec{a}_{t-1}, \vec{x}_{t-1})\prod_{i = 1}^n f_{i,a_{t,i}}\left(x_{t,i}\right).\label{eq:second_S0}\end{equation} Similarly, \begin{equation}f_{\pi\nu'}\left(\left(\vec{a}_t, \vec{x}_t\right)_{t = 1}^T\right) = \prod_{t = 1}^T \pi(\vec{a}_t \mid \vec{a}_1, \vec{x}_1, \dots, \vec{a}_{t-1}, \vec{x}_{t-1})\prod_{i = 1}^n f'_{i,a_{t,i}}\left(x_{t,i}\right).\label{eq:second_S1}\end{equation}

We now return to the logarithm in Equation~\eqref{eq:second_complex_KL}. Based on Equations~\eqref{eq:second_S0} and \eqref{eq:second_S1}, much of the numerator and denominator cancel out, leaving us with \[\log \frac{f_{\pi\nu}\left(\left(\vec{a}_t, \vec{x}_t\right)_{t = 1}^T\right)}{f_{\pi\nu'}\left(\left(\vec{a}_t, \vec{x}_t\right)_{t = 1}^T\right)} = \log \prod_{t = 1}^T \prod_{i = 1}^n\frac{f_{ i,a_{t,i}}\left(x_{t,i}\right)}{f'_{i, a_{t,i}}\left(x_{t,i}\right)} = \sum_{t = 1}^T \sum_{i = 1}^n\log \frac{f_{i, a_{t,i}}\left(x_{t,i}\right)}{f'_{i, a_{t,i}}\left(x_{t,i}\right)}.\]

We can therefore can write the KL divergence as \[D\left(\P_{\pi\nu}, \P_{\pi\nu'}\right) = \sum_{i = 1}^n\sum_{t = 1}^T\E_{\pi\nu}\left[ \log \frac{f_{i,A_{t,i}}\left(X_{t,i}\right)}{f'_{i,A_{t,i}}\left(X_{t,i}\right)}\right].\] Moreover, by the law of total expectation, \[D\left(\P_{\pi\nu}, \P_{\pi\nu'}\right) = \sum_{i = 1}^n\sum_{t = 1}^T\sum_{j = 1}^k\E_{\pi\nu}\left[\left. \log \frac{f_{i,j}\left(X_{t,i}\right)}{f'_{i,j}\left(X_{t,i}\right)}\, \right| \, A_{t,i} = j\right]\P_{\pi\nu}\left[A_{t,i} = j\right].\]

We know that for all $j \not= j^*$, $f_{i,j} = f'_{i,j}$, which means that \begin{equation}D\left(\P_{\pi\nu}, \P_{\pi\nu'}\right) = \sum_{i = 1}^n\sum_{t = 1}^T\E_{\pi\nu}\left[\left. \log \frac{f_{i,j^*}\left(X_{t,i}\right)}{f'_{i,j^*}\left(X_{t,i}\right)}\, \right| \, A_{t,i} = j^*\right]\P_{\pi\nu}\left[A_{t,i} = j^*\right].\label{eq:second_LOTE}\end{equation}
By further conditioning, \begin{align*}&\E_{\pi\nu}\left[\left. \log \frac{f_{i,j^*}\left(X_{t,i}\right)}{f'_{i,j^*}\left(X_{t,i}\right)}\, \right| \, A_{t,i} = j^*\right]\\
=\, &\P_{\pi\nu}\left[X_{t,i} = 0 \mid A_{t,i} = j^*\right]\cdot \log \frac{f_{i,j^*}\left(0\right)}{f'_{i,j^*}\left(0\right)} + \P_{\pi\nu}\left[X_{t,i} = 1 \mid A_{t,i} = j^*\right]\cdot \log \frac{f_{i,j^*}\left(1\right)}{f'_{i,j^*}\left(1\right)}.\end{align*}
Under instance $\nu$, the probability that $X_{t,i} = 0$ given that $A_{t,i} = j^*$ is $f_{i,j^*}\left(0\right)$. Similarly, the probability that $X_{t,i} = 1$ given that $A_{t,i} = j^*$ is $f_{i,j^*}\left(1\right)$. Therefore, \begin{equation}\E_{\pi\nu_{\vec{0}}}\left[\left. \log \frac{f_{i,j^*}\left(X_{t,i}\right)}{f'_{i,j^*}\left(X_{t,i}\right)}\, \right| \, A_{t,i} = j^*\right] =f_{i,j^*}\left(0\right) \log \frac{f_{i,j^*}\left(0\right)}{f'_{i,j^*}\left(0\right)} + f_{i,j^*}\left(1\right)\cdot \log \frac{f_{i,j^*}\left(1\right)}{f'_{i,j^*}\left(1\right)} = D\left(f_{i,j^*}, f'_{i,j^*}\right).\label{eq:second_fi0}\end{equation} Moreover, since $f_{i,j^*}$ is the Bin$\left(\frac{1}{2} + \epsilon\right)$ PMF and $f_{i,j^*}'$ is the Bin$\left(\frac{1}{2} + 2\epsilon\right)$ PMF, we have that $D\left(f_{i,j^*}, f'_{i,j^*}\right) \leq 4\epsilon^2$ for $\epsilon < \frac{1}{5}.$

Combining Equations~\eqref{eq:second_LOTE} and \eqref{eq:second_fi0}, we have that \begin{align*}D\left(\P_{\pi\nu}, \P_{\pi\nu'}\right) &= \sum_{i = 1}^n\sum_{t = 1}^T D\left(f_{i,j^*}, f'_{i,j^*}\right)\P_{\pi\nu}\left[A_{t,i} = j^*\right]\\
&\leq 4\epsilon^2 \sum_{i = 1}^n \E_{\pi\nu}\left[N_{i,j^*}(T)\right]\\
&= 4\epsilon^2\E_{\pi\nu}\left[N_{j^*}(T)\right]\\
&\leq \frac{4\epsilon^2nT}{k-1}.\end{align*}
\end{proof}

Combining Equation~\eqref{eq:BH} and Claim~\ref{claim:second_KL}, and setting $\epsilon = \sqrt{\frac{k-1}{4nT}}$ (in which case $\epsilon < \frac{1}{5}$ for $nT > 7(k-1)$), we have that $R_T(\pi, \nu_{\vec{0}}) + R_T(\pi, \nu_{\vec{1}}) \geq \frac{1}{8e}\sqrt{nT(k-1)}$, so $\max\left\{R_T(\pi, \nu_{\vec{0}}), R_T(\pi, \nu_{\vec{1}})\right\} \geq \frac{1}{16e}\sqrt{nT(k-1)}.$
\end{proof}

\section{Proofs about the Formulation 2 regret bounds}\label{app:2_regret}

In this section, we will use the following notation. For any distributions $\vec{p}_1, \dots, \vec{p}_n \in \cP^{k-1}$, denote the penalty attributed to user $i$ as
 \[P_i((\vecp_i)_{i \in [n]}; \gamma, \eta) = \eta \sum_{j = 1}^k \max\left\{0, \frac{\gamma}{n} \sum_{i' = 1}^n p_{i',j} -  p_{i,j}\right\}.\] The total penalty across all $n$ users is
  \[P((\vecp_i)_{i \in [n]}; \gamma, \eta) = \sum\limits_{i=1}^n P_i((\vecp_i)_{i \in [n]}; \gamma, \eta).\]  
Next, let $\vec{p}_1^*, \dots, \vec{p}_n^* \in \cP^{k-1}$ be distributions that maximize the expected reward minus the penalties. Then the expected regret of a policy $\pi$ under this formulation is \[T\sum_{i = 1}^n \left(\vec{p}_i^*\cdot \vec{\mu}_i - P_i((\vecp_i^*)_{i \in [n]}; \gamma, \eta)\right) - \E\left[\sum_{i = 1}^n \sum_{t = 1}^T \left(X_{i,t} - P_i((\vec{\pi}_i(\vec{h}_{t-1}))_{i \in [n]}; \gamma, \eta)\right)\right].\]

\begin{algorithm}[t]
    \caption{Penalty-UCB (defined by parameter $\delta$)}\label{alg:ucb.f1}
    \begin{algorithmic}[1]
        \Require Failure probability $\delta > 0$
        \State Set  $N_{i, j}(0)= 0,~ \forall i \in [n], j\in [k]$;    $\vec{\hat{\mu}}_i^{(0)} = \vec{0}, ~~ \forall i \in [n]$
        \For{$t \in \{1, \dots, T\}$}
        \If {$t \in \{ 1, \ldots, k\}$}
        \State Set $\vecp_i^{(t)} = \vec{e}_t$
        \Else
        %\State Set $(\vecp_i^{(t)})_{i \in [n]} = \amax\left\{\sum\limits_{i = 1}^n \vecp_i \cdot \vec{\hat{\mu}}_i^{(t-1)} - P_i((\vecp_i)_{i \in [n]}; \gamma, \eta)\right\}$
        \State Set $\left(\vecp_i^{(t)}\right)_{i \in [n]} = \amax\left\{\sum\limits_{i = 1}^n \vecp_i \cdot \vec{\hat{\mu}}_i^{(t-1)}- \eta \sum_{j = 1}^k \max\left\{\frac{\gamma}{n} \sum_{i' =1}^np_{i',j} - p_{i,j}, 0\right\}\right\}$
        \EndIf
        \State Draw $j_i^{(t)} \sim \vecp_i^{(t)} ~~ \forall i \in [n]$
        \State Receive reward $r_i^{(t)} \sim X_{i, j_i^{(t)}}$
        \State $N_{i,  j_{i}^{t}}(t) = N_{i,  j_{i}^{t}}(t-1) + 1, ~~ \forall i \in [n]$
        \State $N_{i,  j}(t) = N_{i,  j}(t-1) , ~~ \forall i \in [n]$ and  $j \neq j_{i}^{t}$
        \State $\beta_{i, j}^{(t)} = \sqrt{\frac{\log(2Tnk/\delta)}{N_{i,j}(t)}}, ~~ \forall i \in [n], j \in [k]$
        \State $\vec{\hat{\mu}}_{i, j}^t = \frac{1}{N_{i, j}(t)}\sum\limits_{\tau = 1}^t r_i^{(\tau)}\ind{j_{i}^{(\tau)} = j} +\beta_{i, j}^{(t)}$, $~~ \forall i \in [n], j \in [k]$
        \EndFor
    \end{algorithmic}
\end{algorithm}

\formtwo*
\begin{proof}
    Fix a timestep $t \in [T]$
    \begin{multline*}
            \sum\limits_{i=1}^n (\vecp_i^{*} \cdot \vec{\mu}_i - \vecp^{(t)}_i\cdot \vec{\mu}_i) -  (P_i((\vecp_i^*)_{i \in [n]}; \gamma, \eta) - P_i((\vecp_i^{(t)})_{i \in [n]}; \gamma, \eta)) \\
              = \sum_{i=1}^n (\vecp_i^*\cdot \vec{\mu}_i  -  P_i((\vecp_i^*)_{i \in [n]}; \gamma, \eta)) - (\vecp_i^{(t)}\cdot \vec{\hat{\mu}}_i^{(t)}  - P_i((\vecp_i^{(t)})_{i \in [n]}; \gamma, \eta)) - \vecp_i^{(t)}\cdot \vec{\mu}_i    + \vecp_i^{(t)}\cdot \vec{\hat{\mu}}_i^{(t)} \\
            \leq \sum_{i=1}^n (\vecp_i^*\cdot \vec{\mu}_i  -P_i((\vecp_i^*)_{i \in [n]}; \gamma, \eta)) - (\vecp_i^{*}\cdot \vec{\hat{\mu}}_i^{(t)}  - P_i((\vecp_i^*)_{i \in [n]}; \gamma, \eta)) - \vecp_i^{(t)}\cdot \vec{\mu}_i   + \vecp_i^{(t)}\cdot \vec{\hat{\mu}}_i^{(t)} \\
            = \sum_{i=1}^n  (\vecp_i^* \cdot \vec{\mu}_i - \vecp_i^*\cdot\vec{\hat{\mu}}_i^{(t)} - \vecp_i^{(t)}\cdot\vec{\mu}_i   + \vecp_i^{(t)}\cdot \vec{\hat{\mu}}_i^{(t)}) \\
        \end{multline*}
    By Claim \ref{clm:ub.f2.1}, $\vecp \cdot \vec{\mu}_i \leq \vecp \cdot \vec{\hat{\mu}}_i^{(t)}$ $\forall i \in [n]$ and all $\vecp \in \R_{\geq 0}^k$
    \begin{align*}
            \sum\limits_{i=1}^n (\vecp_i^{*} \cdot \vec{\mu}_i - \vecp^{(t)}_i\cdot \vec{\mu}_i) -  (P_i((\vecp_i^*)_{i \in [n]}; \gamma, \eta) - P_i((\vecp_i^{(t)})_{i \in [n]}; \gamma, \eta))  &\leq \sum_{i=1}^n (\vecp_i^{(t)} \cdot \vec{\hat{\mu}}_i^{(t)} - \vecp_i^{(t)}\cdot \vec{\mu}_i )  \\
            &\leq \sum_{i=1}^n  \vecp_i^{(t)} \cdot \vec{\beta}_i^{(t)}
        \end{align*}

Thus \[  R_T \leq \sum_{t =1}^T\sum_{i \in [n]}\vecp_i^{(t)} \cdot \vec{\beta}_i^{(t)}\]

Let $\mathcal{F}_{i, t-1}$ be the sigma algebra defined up to the choice of $\vecp_i^{(t)}$ and $j_{i}^{(t)'}$ be a random variable distributed as $\vecp_i^{(t)} \mid \mathcal{F}_{i,t-1}$ and conditionally independent from $j_{i,t}^{(t)}$, i.e.$~ j_{i}^{(t)'} \perp j_{i}^{(t)}\mid \mathcal{F}_{i,t-1}$. Note that by definition the following equality holds: 
\begin{equation*}
    \E_{j_{i}^{(t)}\sim \vecp_{i}^{(t)}}[\beta_{i, j_{i}^{(t)}}] = \E_{j_{i}^{(t)'} \sim  \vecp_{i}^{(t)}}[\beta_{i, j_{i}^{(t)'}} \mid \mathcal{F}_{i, t-1}].
\end{equation*}

Consider the following random variables $A_{i,t} =  \E_{j_{i}^{(t)'} \sim  \vecp_{i}^{(t)}}[\beta_{i, j_{i}^{(t)'}} \mid \mathcal{F}_{i, t-1} ]- \beta_{i,j_i^{(t)}}(t)$. Note that $M_{i,t} = \sum_{s=1}^t A_{i,s}$ is a martingale. Since $|A_t| \leq 2\sqrt{2 \log(nkT/\delta)}$, a simple application of Azuma-Hoeffding    implies that with probability at least $1- \delta$,

\[R_T = \sum_{t =1}^T\sum_{i \in [n]}\vecp_i^{(t)} \cdot \vec{\beta}_i^{(t)} \leq \sum_{t =1}^T\sum_{i \in [n]} \vec{\beta}_{i, j_{i}^{(t)}}^{(t)} + n\sqrt{T\log\left(\frac{nkT}{\delta}\right)\log\left(\frac{1}{\delta}\right)}\]     

Now let us bound  $\sum\limits_{i \in [n]}\sum\limits_{t=i}^T \vec{\beta}_i^{(t)}$, 
\[\sum_{i \in [n]}\sum_{t=i}^T \vec{\beta}_i^{(t)} = \sum_{i \in [n]}\sum_{j \in [k]} \sum_{t=i}^T \beta_{i,j}^{(t)}\ind{j_i^{(t)} = j}  \]
For fixed $i, j$ 
\[ \sum_{t=i}^T \beta_{i,j}^{(t)}\ind{j_i^{(t)} = j} = \sqrt{\log(Tnk/\delta)}\sum_{t=1}^{N_{i,j}(T)}1/\sqrt{t} \leq 2\sqrt{N_{i,j}(T)\log(Tnk/\delta)}\] 
Therefore 
\begin{align*}
    \sum_{i \in [n]}\sum_{t=i}^T \vec{\beta}_i^{(t)} &\leq 2\sum_{i\in [n]}\sum_{j \in [k]}\sqrt{N_{i,j}(T)\log(Tnk/\delta)} \\
    &\leq  2\sum_{i\in [n]}\sqrt{k \sum_{j \in [k]}N_{i,j}(T)\log(Tnk/\delta)} \\
    &=  2\sum_{i\in [n]}\sqrt{k T \log(Tnk/\delta)}  \\
    &= 2n \sqrt{k T \log(Tnk/\delta)} 
\end{align*}
Where the second line follows from the concavity of $\sqrt{\cdot}$ and the penultimate line follows from the fact that $\sum\limits_{j \in [k]}N_{i,j}(T) = T$.
\end{proof}

\section{Proofs about the Formulation 3 regret bounds}\label{app:3_regret}
\ubOptReward*

\begin{proof}
First, for any arm $j\in [k]$, we can exchange the expectation and the maximum in Equation~\eqref{eq:rew3} as follows: \begin{equation}\E_{\pi^*\nu}\left[\max\left\{\frac{\gamma}{n} \sum_{i'=1}^n \hat{p}_{i',j} - \hat{p}_{i,j}, 0\right\}\right] \geq \max\left\{\frac{\gamma}{n}\sum_{i'=1}^n  \E_{\pi^*\nu}\left[\hat{p}_{i',j}\right] -  \E_{\pi^*\nu}\left[\hat{p}_{i,j}\right], 0\right\}.\label{eq:max_ineq1}\end{equation}
Moreover, we can rewrite the expected empirical distribution as follows: \[\E_{\pi^*\nu}\left[\hat{p}_{i,j}\right] = \frac{1}{T}\sum_{t = 1}^T\E_{\pi^*\nu}\left[\textbf{1}_{\{a_{t,i} = j\}}\right] = \frac{1}{T}\sum_{t = 1}^T\E_{\pi^*\nu}\left[\pi^*_i(j \mid \vec{h}_{t-1})\right].\]
Therefore, by Equation~\eqref{eq:max_ineq1}, we have that \[\E_{\pi^*\nu}\left[\max\left\{\frac{\gamma}{n} \sum_{i'=1}^n \hat{p}_{i',j} - \hat{p}_{i,j}, 0\right\}\right] \geq \max\left\{\frac{1}{T} \sum_{t = 1}^T\left(\frac{\gamma}{n}\sum_{i'=1}^n \E_{\pi^*\nu}\left[\pi^*_{i'}(j \mid \vec{h}_{t-1})\right] - \E_{\pi^*\nu}\left[\pi^*_i(j \mid \vec{h}_{t-1})\right]\right), 0\right\}.\]
We can therefore bound $\rew_3(\pi^*, \nu; \eta, \gamma)$ as follows:
\begin{align}&\rew_3(\pi^*, \nu; \eta, \gamma) \nonumber\\
 = &\sum_{i = 1}^n\left(\sum_{t = 1}^T \E_{\pi^*\nu}\left[\vec{\mu}_i \cdot \vec{\pi^*}_i(\vec{h}_{t-1})\right] - \eta \sum_{j= 1}^k \E_{\pi^*\nu}\left[\max\left\{\frac{\gamma}{n} \sum_{i'=1}^n \hat{p}_{i',j} - \hat{p}_{i,j}, 0\right\}\right]\right)\nonumber\\
\leq &\sum_{i = 1}^n\sum_{j = 1}^k\left(\sum_{t = 1}^T \mu_{i,j}\E_{\pi^*\nu}\left[\pi^*_i(j \mid \vec{h}_{t-1})\right]\right.\\
&\left.- \eta \max\left\{\frac{1}{T} \sum_{t = 1}^T\left(\frac{\gamma}{n}\sum_{i'=1}^n \E_{\pi^*\nu}\left[\pi^*_{i'}(j \mid \vec{h}_{t-1})\right] - \E_{\pi^*\nu}\left[\pi^*_i(j \mid \vec{h}_{t-1})\right]\right), 0\right\}\right).\label{eq:exp_inside}
\end{align}

Define the history-independent policy $\vec{p} = \left(\vec{p}_1, \dots, \vec{p}_n\right)$ such that \[p_{i,j} = \frac{1}{T}\sum_{t = 1}^T \E_{\pi^*\nu}\left[\pi_i^*(j \mid \vec{h}_{t-1})\right].\] This is a distribution because for any user $i \in [n]$, \[\sum_{j = 1}^k p_{i,j} = \frac{1}{T}\sum_{t = 1}^T \E_{\pi^*\nu}\left[\sum_{j = 1}^k \pi_i^*(j \mid \vec{h}_{t-1})\right] = \frac{1}{T}\sum_{t = 1}^T \E_{\pi^*\nu}\left[1\right] = 1.\] We rearrange Equation~\eqref{eq:exp_inside} to get that \begin{align*}
  &\rew_3(\pi^*, \nu; \eta, \gamma)\\
  \leq &\sum_{i = 1}^n\sum_{j = 1}^k\left( \mu_{i,j}\sum_{t = 1}^T\E_{\pi^*\nu}\left[\pi^*_i(j \mid \vec{h}_{t-1})\right]\right.\\
  &\left.- \eta \max\left\{\frac{\gamma}{n}\sum_{i'=1}^n \frac{1}{T} \sum_{t = 1}^T\E_{\pi^*\nu}\left[\pi^*_{i'}(j \mid \vec{h}_{t-1})\right] - \frac{1}{T} \sum_{t = 1}^T\E_{\pi^*\nu}\left[\pi^*_i(j \mid \vec{h}_{t-1})\right]\right), 0\right\}\\
    = &\sum_{i = 1}^n\sum_{j = 1}^k\left(T \mu_{i,j}p_{i,j} - \eta \max\left\{\frac{\gamma}{n}\sum_{i'=1}^n p_{i', j} - p_{i,j}, 0\right\}\right).\end{align*} By definition of $\vec{p}^*$, this means that 
    \[\rew_3(\pi^*, \nu; \eta, \gamma) \leq \sum_{i = 1}^n\sum_{j = 1}^k\left(T \mu_{i,j}p^*_{i,j} - \eta \max\left\{\frac{\gamma}{n}\sum_{i'=1}^n p^*_{i', j} - p^*_{i,j}, 0\right\}\right) = \rew_2\left(\vec{p}^*, \nu; \frac{\eta}{T}, \gamma\right).\]

\end{proof}

\lbAnyPolicy*

\begin{proof}
    First, for any arm $j \in [k]$, we can exchange the expectation and the maximum in Equation~\eqref{eq:rew2} as follows: \begin{align*}&\frac{1}{T}\sum_{t = 1}^T\E_{\pi\nu}\left[\max\left\{\frac{\gamma}{n} \sum_{i'=1}^n \pi_{i'}(j \mid \vec{h}_{t-1}) - \pi_{i}(j \mid \vec{h}_{t-1}) , 0\right\}\right]\\
    \geq \, &\E_{\pi\nu}\left[\max\left\{\frac{1}{T}\sum_{t = 1}^T\left(\frac{\gamma}{n} \sum_{i'=1}^n \pi_{i'}(j \mid \vec{h}_{t-1}) - \pi_{i}(j \mid \vec{h}_{t-1})\right) , 0\right\}\right]\\
    \geq \, &\max\left\{\E_{\pi\nu}\left[\frac{1}{T}\sum_{t = 1}^T\left(\frac{\gamma}{n}\sum_{i'=1}^n  \pi_{i'}(j \mid \vec{h}_{t-1}) -  \pi_{i}(j \mid \vec{h}_{t-1})\right)\right], 0\right\}.\end{align*} Using the fact that $\E_{\pi\nu}[\pi_i(j \mid \vec{h}_{t-1})] = \E_{\pi\nu}[\textbf{1}_{\{A_{t,i} = j\}}]$,
 we have that \begin{align}&\frac{1}{T}\sum_{t = 1}^T\E_{\pi\nu}\left[\max\left\{\frac{\gamma}{n} \sum_{i'=1}^n \pi_{i'}(j \mid \vec{h}_{t-1}) - \pi_{i}(j \mid \vec{h}_{t-1}) , 0\right\}\right]\nonumber\\
 \geq \, &\max\left\{\E_{\pi\nu}\left[\frac{1}{T}\sum_{t = 1}^T\left(\frac{\gamma}{n}\sum_{i'=1}^n  \textbf{1}_{\{A_{t,i'} = j\}} -  \textbf{1}_{\{A_{t,i} = j\}}\right)\right], 0\right\}.\label{eq:pi_inside}\end{align}

Next, we use the fact~\citep{Aven85:Upper} that \begin{align}&\max\left\{\E_{\pi\nu}\left[\frac{1}{T}\sum_{t = 1}^T\left(\frac{\gamma}{n}\sum_{i'=1}^n  \textbf{1}_{\{A_{t,i'} = j\}} -  \textbf{1}_{\{A_{t,i} = j\}}\right)\right], 0\right\}\nonumber\\
\geq \, &\E_{\pi\nu}\left[\max\left\{\frac{1}{T}\sum_{t = 1}^T\left(\frac{\gamma}{n}\sum_{i'=1}^n  \textbf{1}_{\{A_{t,i'} = j\}} -  \textbf{1}_{\{A_{t,i} = j\}}\right), 0\right\}\right]\nonumber\\
&- \sqrt{\frac{1}{2}\cdot \text{Var}\left(\frac{1}{T}\sum_{t = 1}^T\left(\frac{\gamma}{n}\sum_{i'=1}^n  \textbf{1}_{\{A_{t,i'} = j\}} -  \textbf{1}_{\{A_{t,i} = j\}}\right)\right)}\nonumber\\
= \, &\E_{\pi\nu}\left[\max\left\{\frac{\gamma}{n} \sum_{i'=1}^n \hat{p}_{i',j} - \hat{p}_{i,j}, 0\right\}\right]- \sqrt{\frac{1}{2T^2}\cdot \text{Var}\left( \sum_{t = 1}^T\left(\frac{\gamma}{n}\sum_{i'=1}^n  \textbf{1}_{\{A_{t,i'} = j\}} -  \textbf{1}_{\{A_{t,i} = j\}}\right)\right)}\label{eq:aven}.\end{align}

Let $Y_t = \frac{\gamma}{n}\sum_{i'=1}^n  \textbf{1}_{\{A_{t,i'} = j\}} -  \textbf{1}_{\{A_{t,i} = j\}}$
and define the martingale difference sequence $D_t := \sum\limits_{\tau = 1}^t(Y_\tau - \E[Y_\tau])$. Then $\text{Var}[D_T] = \text{Var}\left[ \sum_{t = 1}^T Y_t \right]$ and $D_t$ is martingale with bounded increments $|D_{t} - D_{t -1}| \leq 2(\gamma + 1)$. By assumption $\pi_i(t \mid \vec{h}_{t-1}) = 1$ for all $t \leq k$ and $i \in [n]$ so $D_0 = 0$ deterministically. Let $B$ be the event that $|D_T| \leq (\gamma + 1)\sqrt{8T\log T}$. Applying Azuma-Hoeffding for martingales we know  that $\Pr[B^c] \leq \frac{1}{T}.$ Moreover, with probability 1, $|D_T| \leq 2T(\gamma + 1)$. Therefore, by the law of total variance
and Popoviciu's inequality,
\begin{align*}&\text{Var}\left[ \sum_{t = 1}^T Y_t \right]\\
=\, &\text{Var}[D_T]\\
=\,& \text{Var}[D_T \mid B]\Pr[B] + \text{Var}[D_T \mid B^c]\Pr[B^c]\\
&+ \left(\E[D_T \mid B]^2 + \E[D_T \mid B^c]^2 - 2\E[D_T \mid B]\E[D_T \mid B^c]\right)\Pr[B]\Pr[B^c]\\
\leq \,&\text{Var}[D_T \mid B] + \frac{1}{T}\left(\text{Var}[D_T \mid B^c] + \E[D_T \mid B]^2 + \E[D_T \mid B^c]^2 - 2\E[D_T \mid B]\E[D_T \mid B^c]\right)\\
\leq \, &2T(\gamma+1)^2\log T + 17T(\gamma + 1)^2\\
\leq \, & 19T(\gamma+1)^2\log T.\end{align*}

Combining this fact with Equations~\eqref{eq:pi_inside} and \eqref{eq:aven}, we have that

\begin{align}
&\frac{1}{T}\sum_{t = 1}^T\E_{\pi\nu}\left[\max\left\{\frac{\gamma}{n} \sum_{i'=1}^n \pi_{i'}(j \mid \vec{h}_{t-1}) - \pi_{i}(j \mid \vec{h}_{t-1}) , 0\right\}\right]\nonumber\\
\geq\,  &\E_{\pi\nu}\left[\max\left\{\frac{\gamma}{n} \sum_{i'=1}^n \hat{p}_{i',j} - \hat{p}_{i,j}, 0\right\}\right] - \sqrt{\frac{10(\gamma + 1)^2\log T}{T}}.
\label{eq:martingale}
\end{align}
As a result, 

\begin{align*}
&\rew_2\left(\pi, \nu; \frac{\eta}{T}, \gamma\right)\\
=\, &\E_{\pi\nu}\left[\sum_{i = 1}^n \left(\sum_{t = 1}^T \vec{\mu}_i \cdot \vec{\pi}_i(\vec{h}_{t-1}) - \frac{\eta}{T} \sum_{j = 1}^k \max\left\{\frac{\gamma}{n} \sum_{i' =1}^n\pi_{i'}(j \mid \vec{h}_{t-1}) - \pi_i(j \mid \vec{h}_{t-1}), 0\right\}\right)\right]\\
\leq \, &\E_{\pi\nu}\left[\sum_{i = 1}^n \left(\sum_{t = 1}^T \vec{\mu}_i \cdot \vec{\pi}_i(\vec{h}_{t-1}) - \eta \sum_{j = 1}^k \max\left\{\frac{\gamma}{n} \sum_{i'=1}^n \hat{p}_{i',j} - \hat{p}_{i,j}, 0\right\}\right)\right] +  \eta nk (\gamma + 1)\sqrt{\frac{10\log T}{T}}\\
= \, &\rew_3(\pi, \nu; \eta, \gamma)+ \eta nk(\gamma + 1)\sqrt{\frac{10 \log T}{T}}.
\end{align*}

\end{proof}

\empRegret*

\begin{proof}
Let $\vec{p}^*$ be the policy that maximizes $\rew_2\left(\vec{p}, \nu; \frac{\eta}{T}, \gamma\right).$
    We expand the regret as
    \begin{align*}
        &\rew_3(\pi^*, \nu; \eta, \gamma) - \rew_3(\pi, \nu; \eta, \gamma)\\
        =\, &\rew_3(\pi^*, \nu; \eta, \gamma) - \rew_2\left(\vec{p}^*, \nu; \frac{\eta}{T}, \gamma\right) + \rew_2\left(\vec{p}^*, \nu; \frac{\eta}{T}, \gamma\right) - \rew_3(\pi, \nu; \eta, \gamma)\\
        \leq \, &\rew_2\left(\vec{p}^*, \nu; \frac{\eta}{T}, \gamma\right) - \rew_3(\pi, \nu; \eta, \gamma) &\text{(Lemma~\ref{lem:ub_opt_reward})}\\
        \leq \, &\rew_2\left(\vec{p}^*, \nu; \frac{\eta}{T}, \gamma\right) - \rew_2\left(\pi, \nu; \frac{\eta}{T}, \gamma\right) + \eta nk(\gamma+1)\sqrt{\frac{10\log T}{T}}&\text{(Lemma~\ref{lem:lb_any_policy})}\\
        =\, &O\left( n\sqrt{kT\log\left(Tnk\right)} + \sqrt{T\log^2(Tnk)} + \eta nk(\gamma+1)\sqrt{\frac{10\log T}{T}}\right). &\text{(Theorem~\ref{thm:taxation})}
    \end{align*}
\end{proof}

\section{Additional information about the experiments}\label{app:experiments}
\begin{figure}
    \centering
    \begin{subfigure}[b]{0.49\textwidth}
        \centering
        \includegraphics[width=\textwidth]{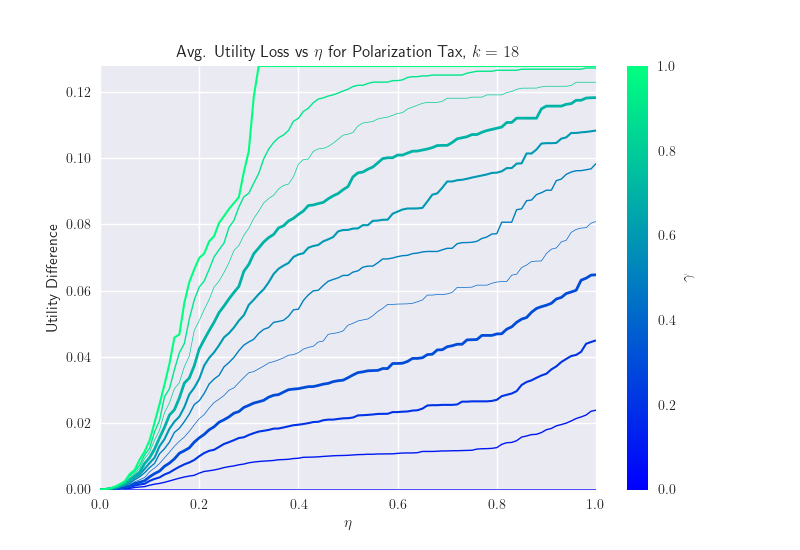}
        \caption{All User Types}
        \label{fig:ptax_util_diff_all}
    \end{subfigure}
    % \hfill
    \begin{subfigure}[b]{0.49\textwidth}
        \centering
        \includegraphics[width=\textwidth]{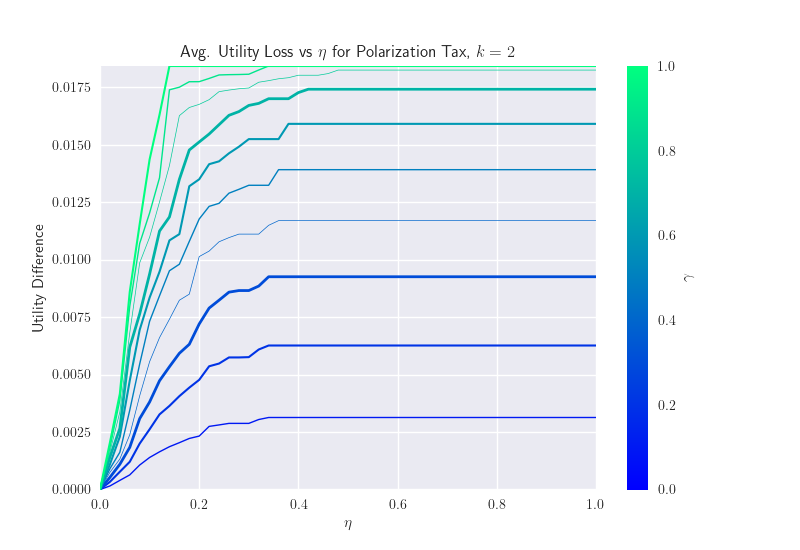}
        \caption{Action and Crime Content Lovers}
        \label{fig:ptax_util_diff_3}
    \end{subfigure}
    \caption{Polarization Tax: Utility Difference as function $\gamma$ and $\eta$}
    \label{fig:ptax_util_diff_diverse}
\end{figure}
Figure~\ref{fig:ptax_util_diff_diverse} plots the change in the additive utility loss \[\frac{1}{n}\sum_{i=1}^n \vec{\mu}_i \cdot \vec{p}_i^* - \frac{1}{n}\sum_{i=1}^n \vec{\mu}_i \cdot \vec{p}_i^{\gamma; \eta}\] for all genres and the entire population of users described in Section~\ref{sec:exp_setup}.

\end{document}